\documentclass{vldb}

\usepackage{pgf,tikz,tikzscale}
\usepackage{ccicons}
\usetikzlibrary{arrows}

\usepackage{amssymb}
\usepackage{epsfig}
\usepackage{xspace}
\usepackage{wrapfig}
\usepackage{microtype}
\usepackage{tabularx}
\usepackage{mathtools}
\usepackage{amsmath}
\usepackage{algorithm}
\usepackage[noend]{algorithmic}
\usepackage{graphicx}
\usepackage{caption}
\usepackage{tcolorbox}
\usepackage{float}
\usepackage{wrapfig}
\usepackage{array}

\graphicspath{{IpeFigures/}{Figures/}}

\newcommand{\denselist}{\itemsep -2pt\parsep=-1pt\partopsep -2pt}

\newcommand{\eps}{\varepsilon}

\newcommand{\Fmk}{\textsf{Approx Hull}\xspace}
\newcommand{\Fmh}{\textsf{Convex Hull}\xspace}
\newcommand{\Fmdp}{\textsf{DP}\xspace}

\newcommand{\Fmg}{\textsf{Gridding}\xspace}
\newcommand{\Fmgd}{\textsf{Grid Kernel}\xspace}
\newcommand{\Fme}{\textsf{Even}\xspace}

\newcommand{\osmdata}{\textsf{OSM}\xspace}
\newcommand{\beijingdata}{\textsf{Beijing}\xspace}

\newtheorem{lemma}{Lemma}[section]
\newtheorem{theorem}{Theorem}[section]
\newtheorem{definition}{Definition}[section]

\newcommand{\R}{\mathbb{R}}

\newcommand{\omt}[1]{}
\newcommand{\etal}{\emph{et al.}\xspace}

\newcommand{\Eu}[1]{\ensuremath{\mathcal{#1}}}

\newcommand{\RangeSet}{\Eu{A}}

\newcommand{\myParagraph}[1]{\vspace{.05in}\noindent {\sffamily\bfseries #1}}

\title{Scalable Spatial Scan Statistics for Trajectories}
\numberofauthors{3}
\author{
	\alignauthor
		Michael Matheny\\
       \affaddr{University of Utah}\\
       \affaddr{Salt Lake City, Utah}\\
       \email{mmath@cs.utah.edu}
\alignauthor
		Dong Xie\\
		\affaddr{University of Utah}\\
       \affaddr{Salt Lake City, Utah}\\
       \email{dongx@cs.utah.edu}
\alignauthor Jeff M. Phillips \titlenote{Jeff Phillips thanks his support from NSF CCF-1350888, ACI-1443046, CNS- 1514520, CNS-1564287, and IIS-1816149. Part of the work was completed while visiting the Simons Institute for Theory of Computing.}\\
		\affaddr{University of Utah}\\
		\affaddr{Salt Lake City, Utah}\\
       \email{jeffp@cs.utah.edu}
}

\begin{document}

\maketitle

\begin{abstract}
We define several new models for how to define anomalous regions among enormous sets of trajectories.  These are based on spatial scan statistics, and identify a geometric region which captures a subset of trajectories which are significantly different in a measured characteristic from the background population.  The model definition depends on how much a geometric region is contributed to by some overlapping trajectory. This contribution can be the full trajectory, proportional to the time spent in the spatial region, or dependent on the flux across the boundary of that spatial region.  Our methods are based on and significantly extend a recent two-level sampling approach which provides high accuracy at enormous scales of data.  We support these new models and algorithms with extensive experiments on millions of trajectories and also theoretical guarantees.  

\end{abstract}


\section{Introduction}
\label{sec:intro}

Large data sets of trajectories have driven much recent research with the goal of understanding the intentions and causes of various correlations hidden within the data.  These data objects are structured to capture movements, interactions, and possibly the intentions of humans and other objects.  Yet, when one considers these trajectories in unison, they usually appear as just a tangled mess.  And as the data grows (e.g., millions of objects with billions of data points), this task does not seem to aggregate and statistically simplify, rather it just becomes more unwieldy and unmanageable.

We consider new methods that use the underlying geometry of the trajectories to identify statistically significant spatial anomalies.  
Critically, these models deviate from density-based models (like DBScan~\cite{EKSX96}) which would only identify populous regions (of course a lot of traffic exists in New York or Beijing!)
Rather, our new models and algorithms identify geometric regions where some labeled aspect (e.g., trajectories of sick people) significantly deviates from what is expected or is in contrast to the trajectories of the background population.  And our approaches work at enormous scale -- required for modern large data sets, and for the statistics to be meaningful.  

Developing such models for trajectories, comparing to a background population, is highly motivated: identifying significant population or demographic shifts, pinpointing the likely location responsible for disease due to prolonged exposure among a dynamic population, or geolocating a nefarious wifi access point affecting cell phones which transiently pass by.  There are no existing mechanisms for addressing some of these goals specific to trajectory data, and certainly not for massive data sets. 


When base objects are points (not trajectories), such comparative anomaly tasks are typically resolved with a \emph{Spatial Scan Statistic} (SSS)~\cite{Kul97}.  This is one of the most common tasks within Geographic Information Science, with applications to detecting hotspots with elevated levels of disease, crime, or demographic traits~\cite{Kul97,OJPHI7599,NM04,Tango2005,ACTW18}.  These identify a region with maximum log-likelihood score $\Phi$, out of a large prescribed family of regions $\Eu{C}$, and have been shown empirically and theoretically to have high statistical power~\cite{Kul97,ACTW18}.  But this search of \emph{all} regions $C \in \Eu{C}$ (usually associated with a geometric family of shapes) is computationally onerous, and the most common software SatScan~\cite{Kul7.0} is only able to scale by restricting the class of regions to ones specifically chosen by the user. 
Moreover, there are only a few limited extensions towards trajectory data~\cite{PCLZ2011,LZCYX11} and these rely heavily on heuristic aggregation.

%

\myParagraph{Our Contributions.}
We introduce three new models to allow for geometric analysis of trajectory data. These models (derived in Section \ref{sec:model}) identify geometric regions where many trajectories of interest have passed through (full model), spent time in (partial model), or began or ended in (flux model).  These models are new and well-motivated, but they have not before been a computationally feasible objective to consider at scale.

We design sampling and scanning algorithms (in Section \ref{sec:algo}) that allow for extremely scalable methods for identifying anomalous patterns captured in the above geometric models.  
Our methods are not limited by the data size, instead they depend on the error---spatial and combinatorial---that researchers are willing to approximate the final statistical quantity to. 
The scalability as well as statistical power of these models are demonstrated on enormous data sets containing up to several million trajectories with over 1 billion waypoints (in Section \ref{sec:exp}).  

Of the three models, we develop a reduction for the \emph{flux} and \emph{partial} model to a scalable point-based scanning framework (see Section \ref{sec:algo}).  However for the third model, \emph{full}, such reductions are not possible, and considerable new scanning and sampling mechanisms are developed:  
Section \ref{sec:coresets} designs compact coresets for each trajectory that preserves spatial guarantees and converts trajectories to labeled point sets.  
Section \ref{sec:VC} provides trajectory specific sampling theorems (VC-dimension bounds for trajectory-based range spaces) necessary to apply the underlying sample approach.  
Section \ref{sec:full-scanning} develops algorithms that use new data structures to scan the data resulting from trajectory coresets and their samples in a scalable way.  
The overall accuracy and runtime bounds for of our sampling, coreset, and scanning algorithms are ultimately summarized in Table \ref{tab:alg}.  
\section{Preliminaries and Overview}
\label{sec:prelim}

We begin with an overview of the mathematical modeling, geometric, statistical, and algorithmic preliminaries to frame our new contributions.

\myParagraph{Trajectory models.}
\label{sec:trajectories}
We model a trajectory $t$ via waypoints $P_t = \langle p_1, \ldots, p_m \rangle$ as the tracing out of the ordered sequence of connected segments $s_1, \ldots, s_{m-1}$, with $s_j = \overline{p_j p_{j+1}}$.  Trajectory $t$ has total arclength $L(t)$ (or when $t$ is clear, just $L$).  Computationally and also as a way of defining regions, it is usually necessary to consider trajectories via just some set of ordered waypoints  $p_1, \ldots, p_m$ constructed via one of the methods discussed in Section \ref{sec:coresets} instead of as an ordered sequence of connected segments. 
For the partial model the parameterization of the trajectories will be important, and we use arclength by default.   We could alternatively use a time-based parametrization, but we otherwise do not explicitly modeling timing information in this paper, leaving this potential extension to future work.


\myParagraph{Range spaces.}
\label{sec:RS}
To study spatial anomalies applied to trajectories we need a way to model how or when a trajectory interacts or intersects with a potential region of interest.  
To do this we review the definition of a \emph{range space}.  A range space is a pair $(X,\Eu{A})$ consisting of a set of objects $X$ (the \emph{ground set}) and a set $\Eu{A} \subset 2^X$ of subsets of $X$ (the \emph{ranges}), that is a subset of all possible subsets $2^X$. 
Range spaces are essential objects in both data structures (e.g., for range searching~\cite{AE99}) and machine learning (for sample complexity of learning~\cite{VC71,HW87}).  For example, classically let $X$ be points in $\R^2$, and $\Eu{A}_\Eu{C}$ be the subsets induced by intersection with a disk $C \in \Eu{C}$. 

For this paper, we will define new families $\Eu{A}_\Eu{C}$ induced by a set of shapes $\Eu{C}$, focusing on those induced by halfspaces, disks, or axis-aligned rectangles.  In particular, we are interested when the ground set is a set of trajectories $T$ (or derived appropriately from $T$).   The new models (we will define in Section \ref{sec:model}) will specify the definition of intersection $T \cap C$, for prescribed spatial regions $C \in \Eu{C}$, to induce a set $A_C \in \Eu{A}_\Eu{C}$.  That is each induced subset of trajectories $A_C \in \Eu{A}_\Eu{C}$ corresponds with a potentially anomalous region of interest.  


\myParagraph{Statistical discrepancy.}
In all of the forthcoming models, each trajectory $t \in T$ has a recorded value $r(t)$ and baseline value $b(t)$.  It is typically sufficient to consider \emph{all} trajectories $t \in T$ have $b(t) = 1$ (that is they are all part of an observed background population) and that $r(t) \in \{0,1\}$ where the set $T_r = \{t \in T \mid r(t) =1\}$ are the trajectories of interest -- although these can change in accordance with objective function $\phi$ in a variety of statistical settings~\cite{Kul7.0,APV06}.   
Even when we study parts of trajectories and segments of trajectories, these traits (especially the recorded value $r(t)$) is held fixed for the entire trajectory.  

Then let $r(C) = \frac{|T_r \cap C|}{|T_r|}$ (resp. $b(C) = \frac{|T \cap C|}{|T|}$) be the fraction of all recorded trajectories (resp. all trajectories) within the range $C$.  Modeling $r$ as Poisson, the log-likelihood of the recorded rate being different than the baseline rate is $\Phi(C) = \phi(r(C),b(C))$, where 
$\phi(r,b) =r \log \frac{r}{b} + \left(1 - r\right) \log \frac{1-r}{1 - b}.$
However, in general, we can computationally reduce to a simpler 
subproblem~\cite{APV06,MP18a} with a ``linear'' model e.g., 
$\phi(r,b) = |r-b|$.  
Ultimately, the SSS is $\max_{C \in \Eu{C}} \phi(C)$.  It uses the magnitude of
$\Phi(C^*)$ (and permutation testing) to determine if the most anomalous region $C^*$ is
statistically significant.  
We will specifically be interested in an approximate variant:

\begin{definition}[$\eps$-Approx. Spatial Scan Statistic]
Consider a range space $(X,\Eu{A}_\Eu{C})$ defined by a family of shapes $\Eu{C}$, and a discrepancy function $\Phi : \Eu{C} \to \R$.   
An \emph{$\eps$-approximate spatial scan anomaly} is a shape $\hat C \in \Eu{C}$  so that 
\[
\Phi(\hat C) + \eps > \Phi(C^*)
\] 
where $C^* = \arg\max_{C \in \Eu{C}} \Phi(C)$.  
Then the corresponding \emph{$\eps$-approximate spatial scan statistic} is $\Phi(\hat C)$.  
\end{definition}

\begin{figure}
\centering
\includegraphics[width=.98\linewidth]{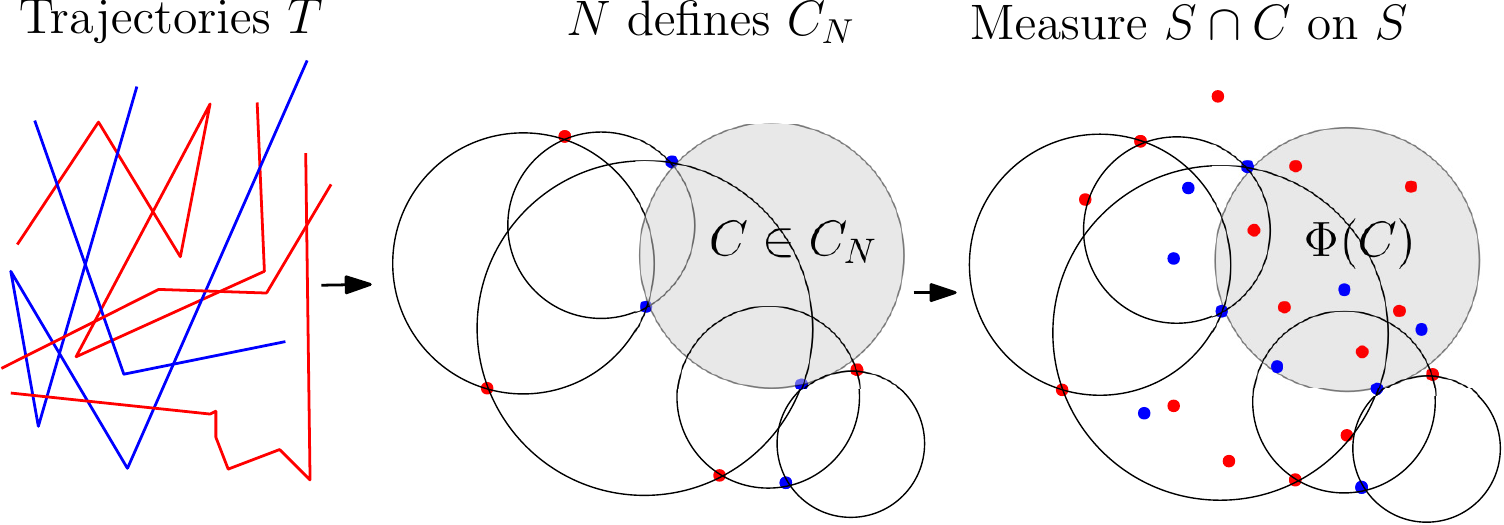}
\vspace{-3mm}
\caption{\label{fig:sss}
A sparse sample of the data defines a set of regions and these regions can be efficiently measured on a larger sample that maintains the density of the original data set.}
\vspace{-2mm}
\end{figure}

\myParagraph{Two-level sampling.}
\label{sec:overview}
The algorithmic goal in this paper is to find an $\eps$-approximate spatial scan anomaly for trajectory range spaces (with forthcoming definitions in Section \ref{sec:model}).  
Our scanning algorithms are based on recent work for calculating $\eps$-approximate spatial scan statistics on simple geometric shapes over point sets at scale~\cite{SSSS,MP18b}.  The main ideas are to construct a \emph{two-level sampling} of a large data set into sets $S$ and $N$; see Figure \ref{fig:sss}.  The larger ``sample'' set $S = S_r \cup S_b$ (where $S_r$ is the recorded set and $S_b$ is the baseline set) is used as proxy for the density of the data in any range, and the smaller ``net'' data sets $N$ is used to define the combinatorial set of ranges we will consider.  That is, $N$ \emph{defines a subset of shapes} $C \in \Eu{C}_N$; those which include combinatorial distinct subset of points in $N$, and then $S$ is used to \emph{estimate the $\Phi(C)$ value} via $C \cap S$.   
To achieve $\eps$ error in $\Phi(C)$, we (roughly) need to set $n = \nu/\eps$ and $s = \nu/
(2\eps)^2$, where $\nu$ is the VC-dimension of the range space~\cite{MP18b}.  
Then there exist shape specific methods to quickly scan and evaluate the ranges (induced by $N$) and values (from $S$)~\cite{MP18b}.  

\myParagraph{Remaining challenges.}
\label{sec:challenges}
Several significant tasks remain to adapt this framework to the new trajectory scanning models defined next.  
First, we require to spatially approximate the trajectories -- the raw trajectories are either unnecessary or too difficult to work with; we describe our methods in Section \ref{sec:coresets}.  
Second, we need to formalize the definitions of regions $\Eu{C}_N$ from the set $N$, and bound the VC-dimension $\nu$ with respect to the resulting range spaces; we do so in Section \ref{sec:VC}.  
Next, for the partial and flux models, we devise reductions to point set variants in Section \ref{sec:reduction}.  Given these reductions, we can invoke existing fast scanning algorithms~\cite{MP18b}.  
However, for the full model such reductions are not possible, and we need to develop new methods to quickly iterate over all these ranges on $S$; that is how to quickly ``scan" over the net.  This varies with the geometric properties of the shapes; we focus on halfspace, disk, and rectangle ranges, and devise new scanning methods for each of them in Section \ref{sec:full-scanning}.




\section{New Trajectory Range Models}
\label{sec:model}

We introduce three new models of how to define range spaces when applied to trajectory-derived ground sets and geometric ranges which define the region of interest.  These models capture:
(i) regions with a high percentage of measured trajectories passing from inside to outside (the \emph{flux model}), (ii) regions with a high percentage of the total arclength of measured trajectory data (the \emph{partial range model}), and (iii)
regions with a high percentage of measured trajectories pass through them (the \emph{full range model}).

\begin{figure}[hb]
\includegraphics[width=\linewidth]{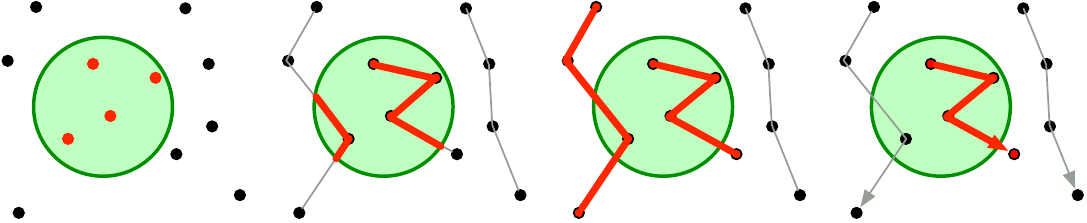}
\vspace{-3mm}
\caption{Models of how $3$ trajectories intersect (in red) with circular shape (in green).  From left to right: point-based model; partial range model; full range model; flux model.}

\label{fig:models}
\end{figure}

\subsection{Flux Model}
\label{sec:flux}

The simplest version of this problem is the flux model.  We search for a shape $C \in \Eu{C}$ where a proportionally high number of trajectories start inside the shape $C$, and end outside of the shape $C$, or vice versa.

\myParagraph{Mathematical Definition.}
In this setting, we can reduce each trajectory to two waypoints: the first and the last.  To satisfy a range intersection, the first must be inside and the last must be outside (or vice-versa with opposite effect).  That is we define two sets $X^{b} = \{p_1 \mid t \in T, t = \langle p_1, \ldots, p_m \rangle \}$ (the beginning set) and $X^{e} = \{p_m \mid t \in T, t = \langle p_1, \ldots, p_m \rangle \}$ (the end set).
Then we attempt to find the range where $\arg \max_{C \in \Eu{C}} |C \cap X^{e}| - |C \cap X^{b}|$ or vice-versa. Since we only care about the endpoints of trajectories we can directly reduce this problem to highly-scalable techniques for point-based algorithms  (see Section \ref{sec:reduction}).  


\myParagraph{Motivating scenarios.}
This model arises when finding a boundary that differentiates two types of traffic.  For instance say a city would like to place a toll fee that mostly affects tourists (using rental cars as proxy) versus locals (using other GPS databases as proxy).  The shape boundaries with high flux are potential good choices.  Alternatively, if trajectories document addresses of people over time (for instance the Utah Population Database which tracks addresses for 50 years), this can be used to identify significant migratory patterns of parts of the population.  Regions of unusually high flux or any fixed window of time may also be useful for managing crowd control in packed sporting or music events.

\begin{figure}
\begin{tcolorbox}[colback=white]
\begin{center}
\textbf{Flux Model Example}
\end{center}
\vspace{-2mm}
\includegraphics[width=\linewidth]{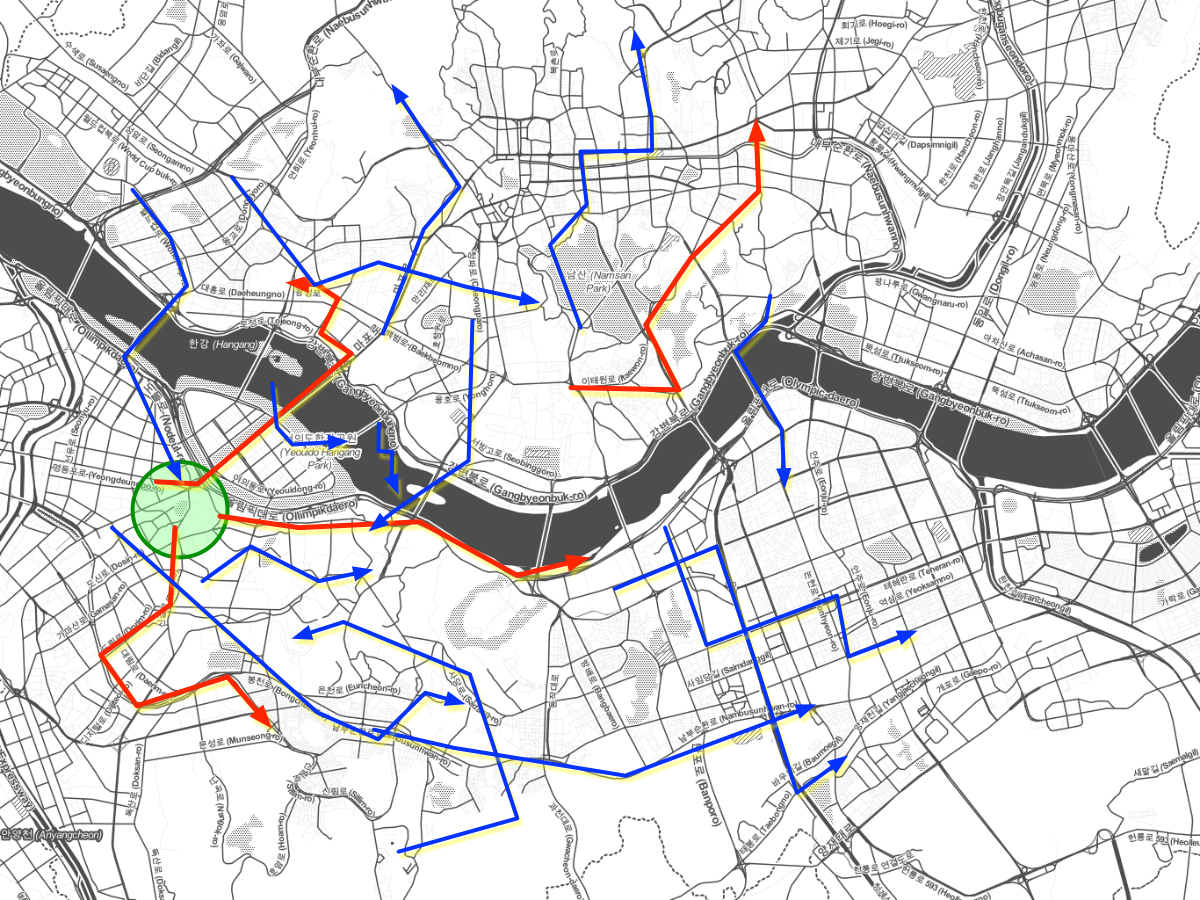}

In the image above, let the trajectories represent fair paths of a Kakao Taxi driver in Seoul over a day.  
The red trajectories represent significant tips, and the blue ones all other fairs paths.  
Identifying a region (in green) where red (high tip) routes started, and left from, without too many blue routes of the same form, will indicate a good place to try to find profitable customers.  This would correspond with the disk maximizing $\Phi$ under the flux model. \footnotemark[\the\numexpr\value{footnote}+1] 
\end{tcolorbox}
\end{figure}

\subsection{Partial Range Model}
\label{sec:partial}

In the partial range model, we want to find a shape $C \in \Eu{C}$ where the weight a trajectory contributes to $C$ is proportional to the normalized length of the trajectories inside. A trajectories intersection with a shape $C \in \Eu{C}$ is fractional; specifically $\mu(t \cap C)$ is the fraction of the total arclength of $t$ within $C$.  That is if $1/3$ of a trajectory $t$ intersects $C$, then $\mu(t \cap C) = 1/3$ or if $1/50$th of the total length of all trajectories intersect $C$ then $\mu(T \cap C) = 1/50$. The contribution $\mu(t \cap C)$ depends on the parametrization of $t$.  This could be by arclength, by time, by fuel used, or any other quantity.  In this paper we simply use arclength in our experiments.  

\begin{figure}
\begin{tcolorbox}[colback=white]
\begin{center}
\textbf{Partial Range Model Example}
\end{center}
\vspace{-2mm}
\includegraphics[width=\linewidth]{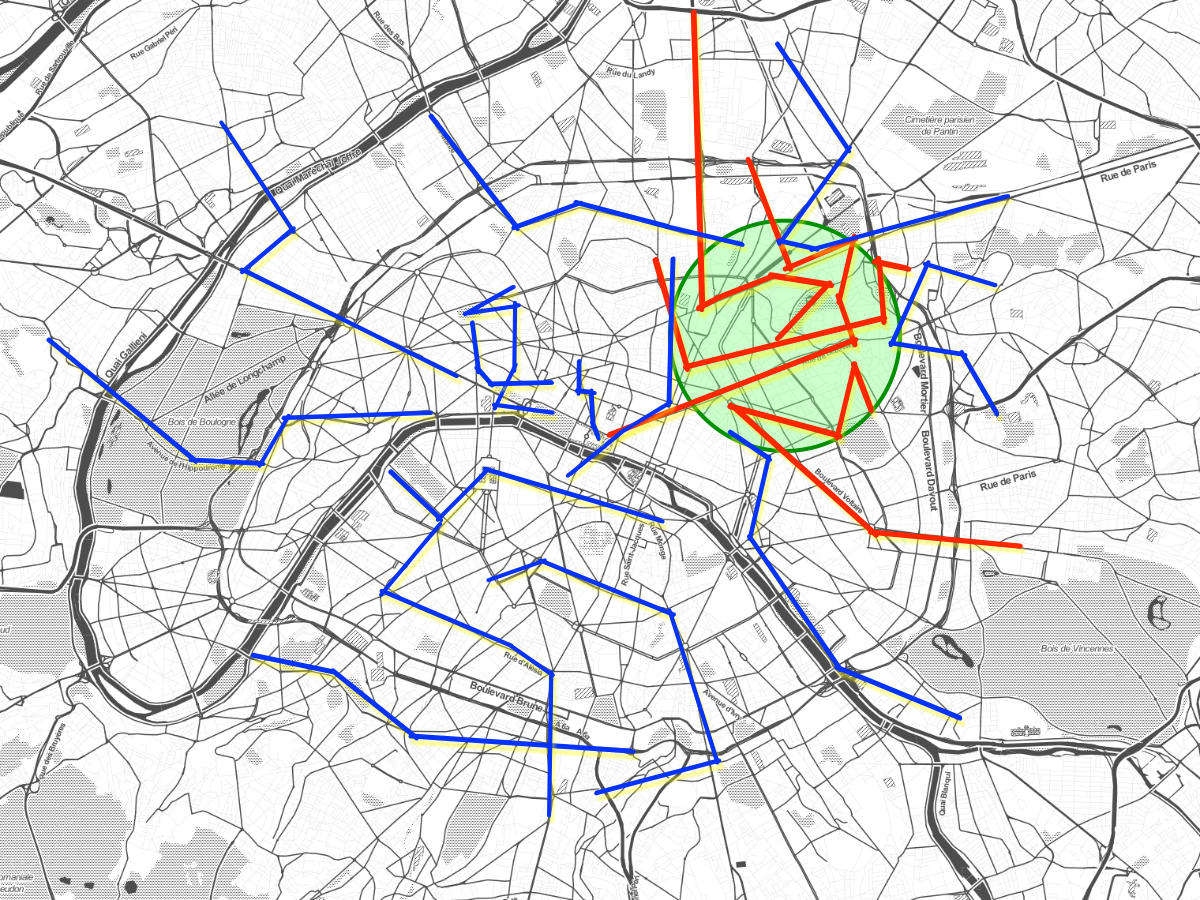}

Let the image above model the daily commutes of patients that have checked into Paris hospitals.  The red curves correspond with patients who reported a serious lung and stomach symptoms when they arrived to work, consistent with exposure to a dangerous chemical.   The blue trajectories are other patients with more standard and mainly benign symptoms.  A region where more red trajectories spent significant time, and blue trajectories did not spend much time, would help pinpoint the region with largest potential to be the source of these symptoms.  And finding such region corresponds with the disk that maximizes $\Phi$ under the partial range model. 
\footnotemark[\the\numexpr\value{footnote}+1] 
\end{tcolorbox}
\end{figure}

\myParagraph{Mathematical Definition.}
More formally, the set $C \cap T$ corresponds to the set of points comprising partial trajectories which are inside $C$.  The true ground set $X$ is then $X = \{ x \in \R^2 \mid x \in s_j \in t \in T \}$.  That is $X$ is a subset of $\R^2$ so any $x \in X$ lies on some $t \in T$.  This corresponds to an infinite (not combinatorially defined) range space~\cite{VC71} as $(X,\Eu{A}_\Eu{C})$ where $\Eu{A}_{\Eu{C}} = \{ A_C = X \cap C \mid C \in \Eu{C}\}$.  
A random sample of points from this infinite point set $X$ falls into the existing theory from \cite{MP18b,SSSS}. Therefore existing scanning algorithms can be directly applied, even ignoring the relationship between these sampled points and the trajectories, since the ground set $X$ is a subset of $\R^2$.  

\myParagraph{Motivating scenarios.}
The partial range model is important to help automatically identify regions which statistically lead to some measured characteristic, proportional to how long the object generating the trajectory spends in that region.  For instance, consider a mysterious sickness that health officials suspect is tied to prolonged exposure to some chemical event.  By finding a compact region where many of the inflicted people spend a considerable amount of time, compared to all in that region, this will provide a candidate location for the epicenter of that exposure.  
Alternatively, consider a measured set of unprofitable (or highly profitable) taxi/Uber drivers;  can we identify regions of a city where they spend a proportionally higher percentage of their time.  These and similar scenarios are directly modeled by the partial range scan statistic, and hence demand scalable solutions.

\subsection{Full Range Model}
\label{sec:full}

In the full range model, we seek to find the shape $C \in \Eu{C}$ which intersects the most trajectories of interest, compared to some baseline set of trajectories. A trajectory's contribution to $C$ is \emph{not} proportional to intersection size, but instead all or nothing.

\begin{figure}
\begin{tcolorbox}[colback=white]
\begin{center}
\textbf{Full Range Model Example}
\end{center}
\vspace{-2mm}
\includegraphics[width=\linewidth]{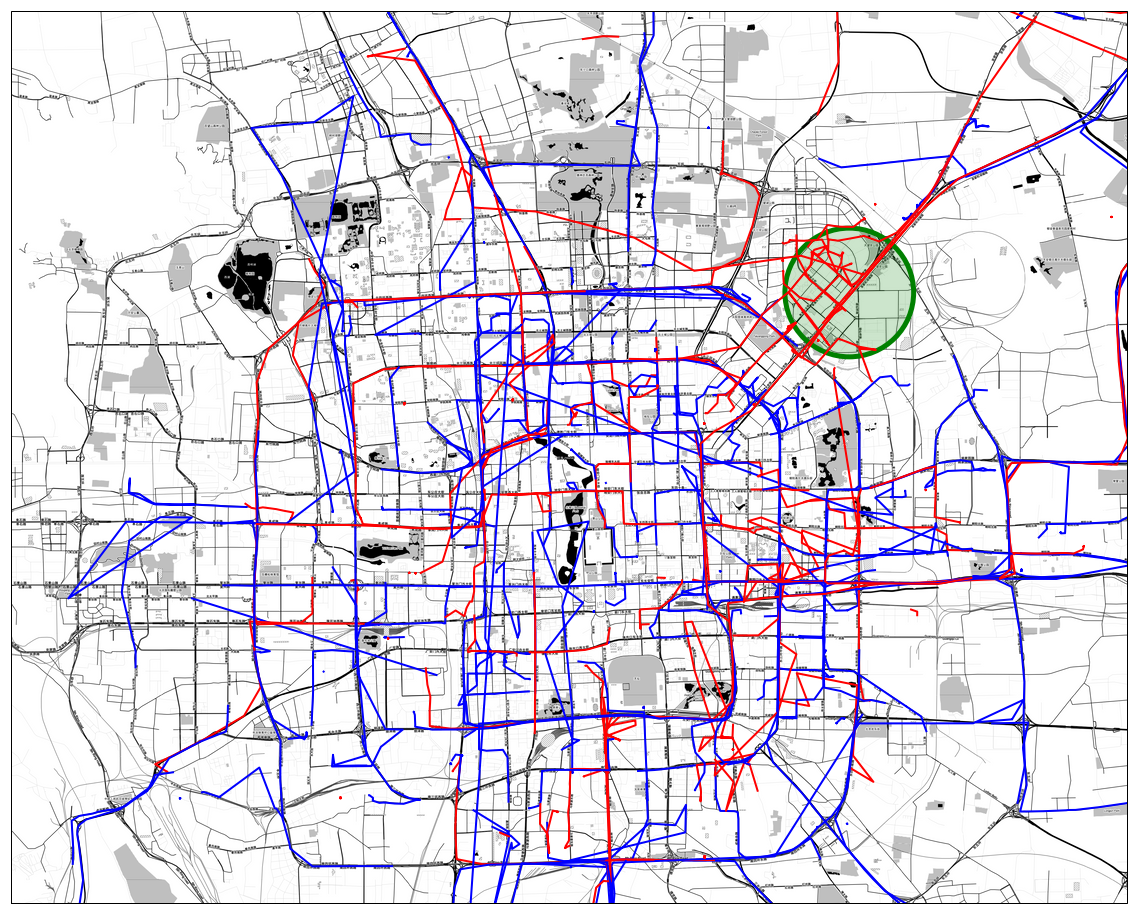}

Consider the set of limousine trajectories from a day in Beijing in the image above.   The red trajectories may represent taxis which reported navigation systems malfunctioning, while the blue ones did not.  You suspect a jamming devise, with a fixed radius, was placed by a competitor.  Detecting the circular shape (in green) which maximizes the fraction of red trajectories which pass through, while minimizing the fraction of blue, is the most likely model for where a jamming devise could have been located.  This corresponds precisely to finding the disk with largest $\Phi$ value under the full range model.  
\footnotemark[\the\numexpr\value{footnote}+1] 
\end{tcolorbox}
\end{figure}
\refstepcounter{footnote} \footnotetext{ Map tiles by Stamen Design, under CC BY 3.0. Data by OpenStreetMap, under ODbL.}

\myParagraph{Mathematical Definition}
The contribution of a trajectory $t \in T$ in a shape $C \in \Eu{C}$ is binary.  It is $1$ if there is any point where the trajectory enters the shape, and is only $0$ if the trajectory is never inside the shape.  And as such ranges $A_C \in \Eu{A}_\Eu{C}$ are defined $A_C = \{t \in T \mid t \text{ intersects } C\}$, for some shape $C \in \Eu{C}$.  
There is no need to parameterize the trajectories in this scenario.

We will typically simplify trajectories by creating a small set of labeled points $P'_t$ to represent them.  In this scenario, we say $t$ intersects shape $C$ if \emph{any} $p'_j \in P'_t$ intersects $C$.

Trajectory simplification is in fact necessary for several reasons: trajectories with an unbounded number of waypoints do not satisfy VC dimension based approximation guarantees (see Section \ref{sec:VC}), for certain range spaces there is no obvious way to define a set of regions to enumerate when considering line segments directly, and it is computationally easier than checking for full intersection with $t$ in most scenarios. Unfortunately this reduction does not allow direct application of some of the previous approaches in \cite{MP18b}, since it needs to only count a trajectory as intersecting a range once even if multiple points are inside. 

\myParagraph{Motivating scenarios.}
This model arises when just the fleeting intersection with a spatial region is enough to trigger a measured event for that entire trajectory.  Consider a set of cars with slow leaks from nails in their tires (found at the end of the day); a region many of them passed through would more likely be someplace with nails on the road.  Or consider a set of people's cell phones which a virus, suspected to be infected when they pinged some wifi access point; then just passing near that access point may be enough to trigger the event, and finding a region with high-density of cases of the virus would provide a probable location of the offending access point.  
Or consider tracking a set of animals (e.g., cows, migrating birds) where a subset become sick; then a full range spatial anomaly may indicate a contaminated watering hole.


\section{Spatial Approximation}
\label{sec:coresets}

As described, for a shape $C \in \Eu{C}$ and trajectory $t$, the critical operation in the full model is determining $t \cap C = \emptyset$.  This is far more efficient if we can approximate $t$ by a set of $k$ approximate waypoints $P'_t = \{p_1', p_2', \ldots, p_k'\}$, and then use $P'_i \cap C$ as a proxy for $t \cap C$.  The critical aspects of such an approximation is to keep the size of the approximation $k$ small, even for long trajectories, and to ensure that the answer to the intersection between a shape $C$ and the point set $P'_t$ approximates $C \cap t$.   

Specifically, we desire an \emph{$\alpha$-spatial approximation} (or just \emph{$\alpha$-approximation} for short).  For a trajectory $t$ and any range $C \in \Eu{C}$, we say $P'_t$ is an $\alpha$-spatial approximation of $t$ under two conditions (see Figure \ref{fig:a-apx}(Left)): 
\begin{enumerate} \denselist
\item (no false positives) If $t$ does not intersect $C$, then $P'_t$ does not intersect $C$.  
\item (limited false negatives)  If for \emph{all} unit vectors $v$, $t + v \alpha$ intersects $C$, then $P'_t$ intersects $C$; here $t+ \alpha v$ is a shift of the entire trajectory in a direction $v$ by $\alpha$.  
\end{enumerate}

\begin{figure}[b]
\centering
\includegraphics[width=.6\linewidth]{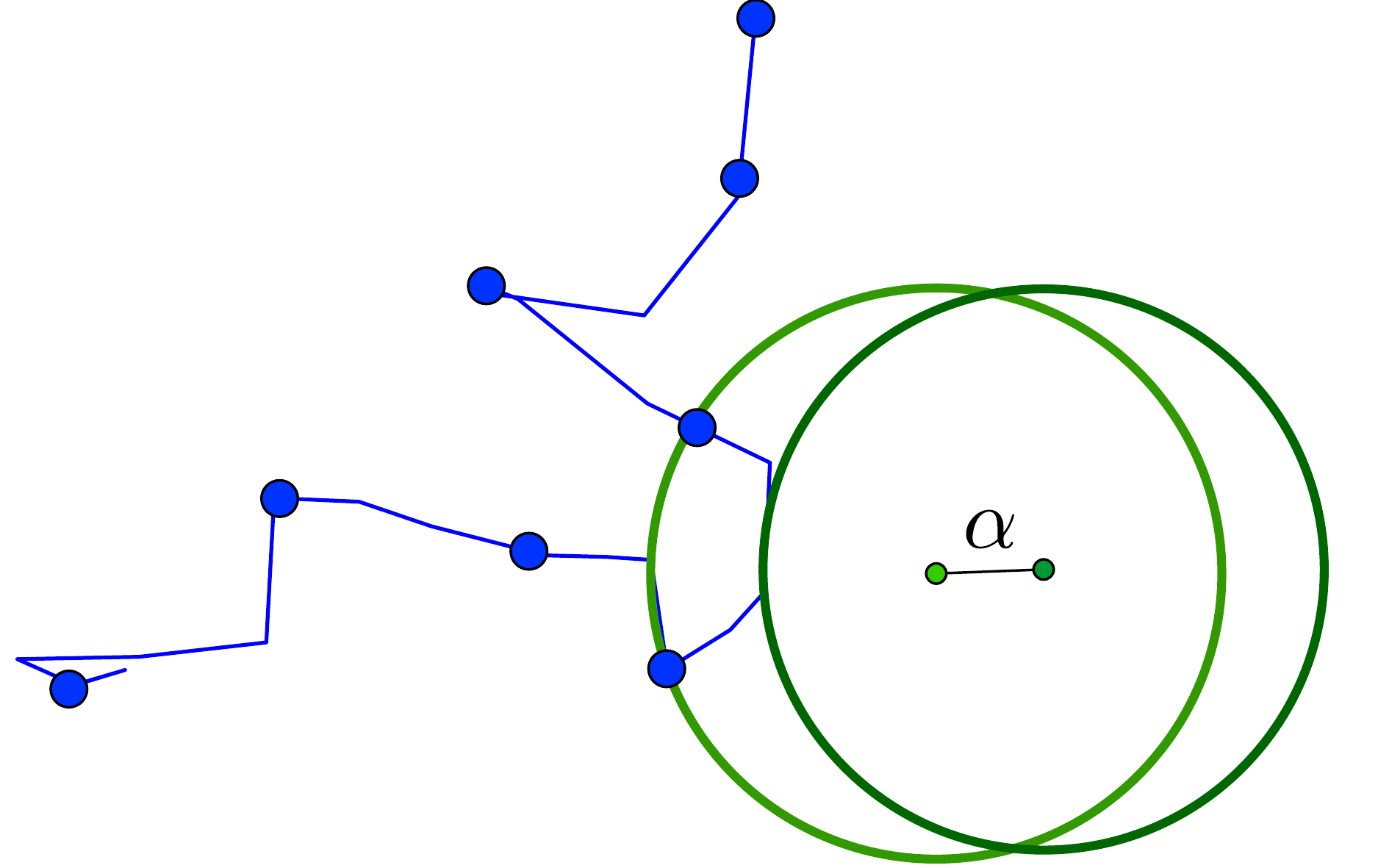}
\includegraphics[width=0.38\linewidth]{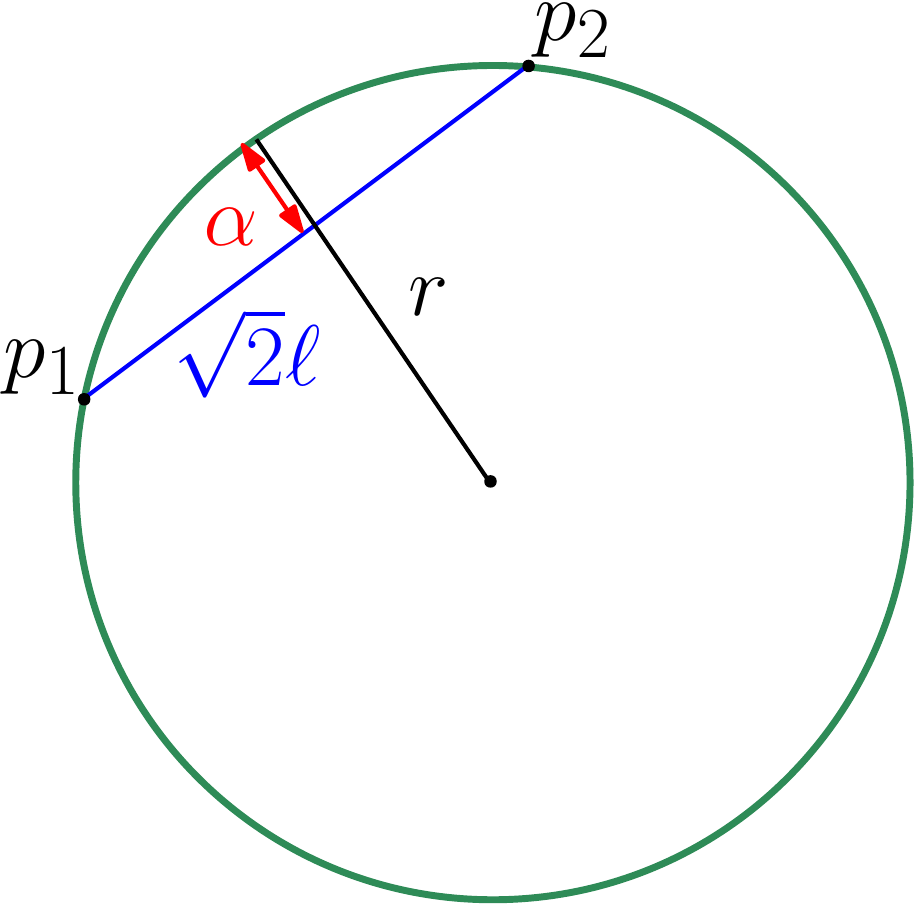}

\caption{\label{fig:a-apx} \label{fig:archeight}
Left:  Converting a trajectory into a set of points which approximate its shape preserves an $\alpha$-spatial guarantee, since we can find a nearby disk that still intersects the trajectory.
Right: The relation between $\alpha$, $\gamma$, and $r$ in proof of Lemma \ref{lem:gk-err} for some line segment $\overline{p_1 p_2}$ on the boundary of one of our kernels.}
\end{figure}

\begin{figure*}[t]
	\includegraphics[width=0.19\linewidth]{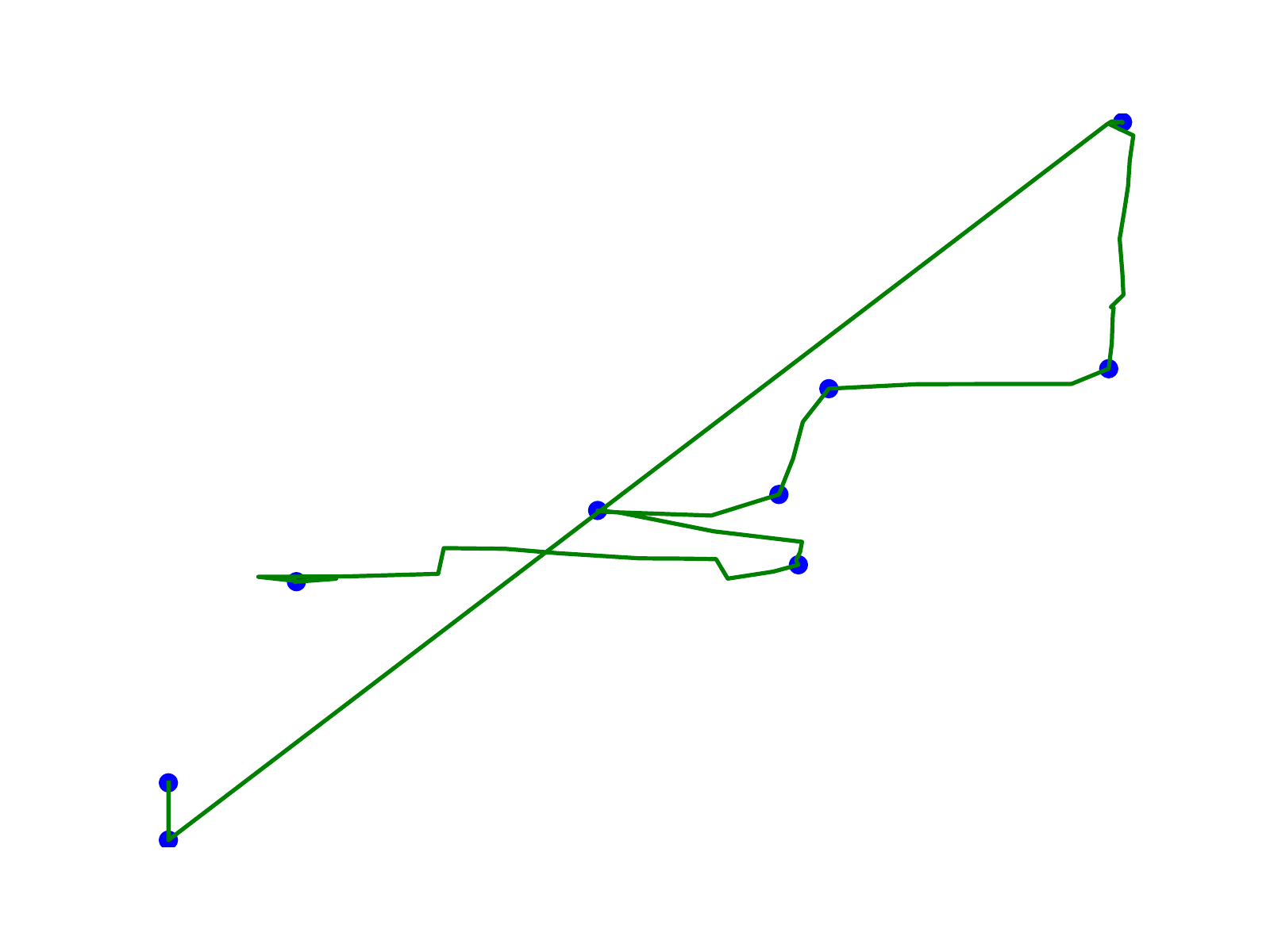}
	\includegraphics[width=0.19\linewidth]{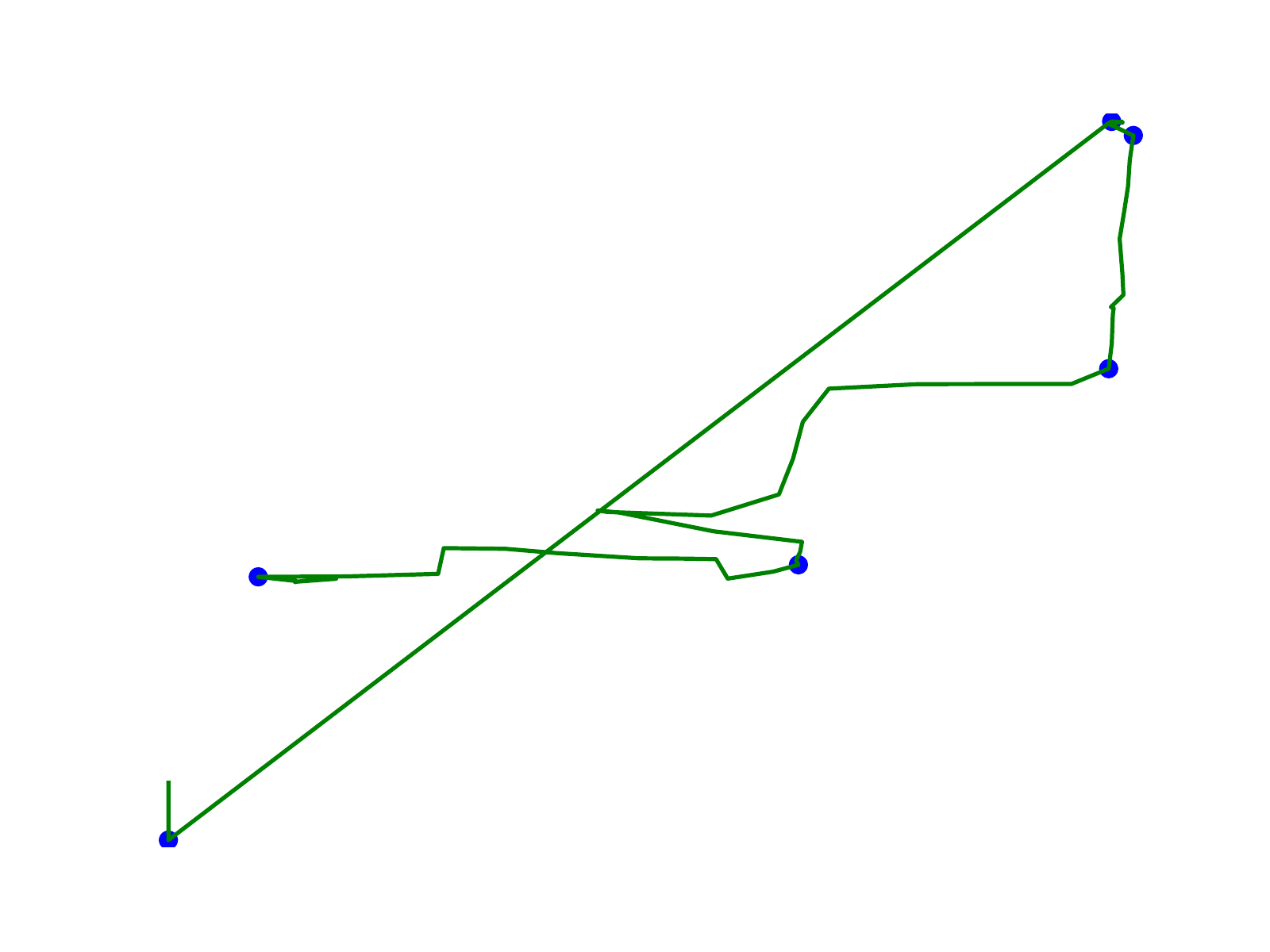}
	\includegraphics[width=0.19\linewidth]{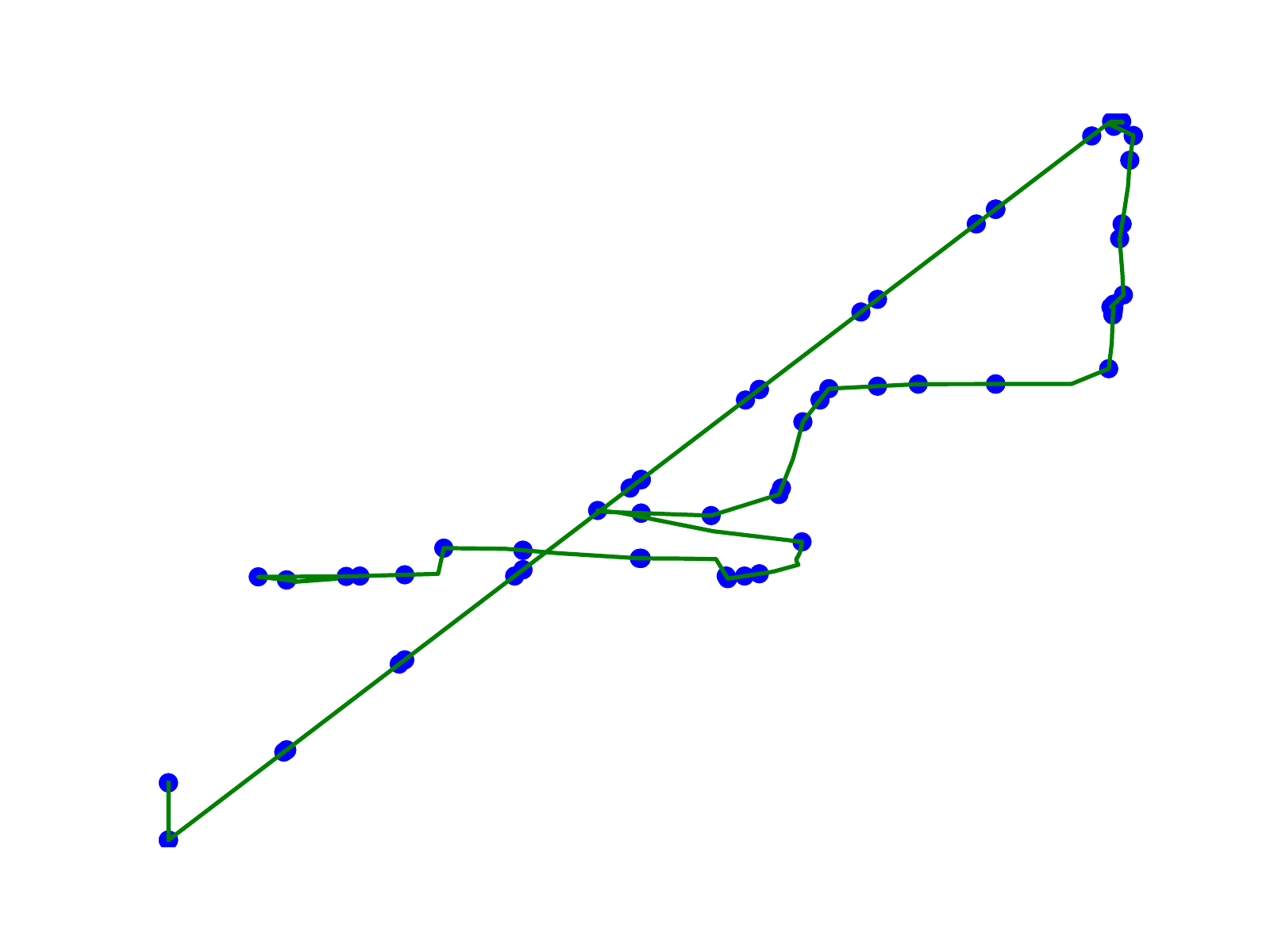}
	\includegraphics[width=0.19\linewidth]{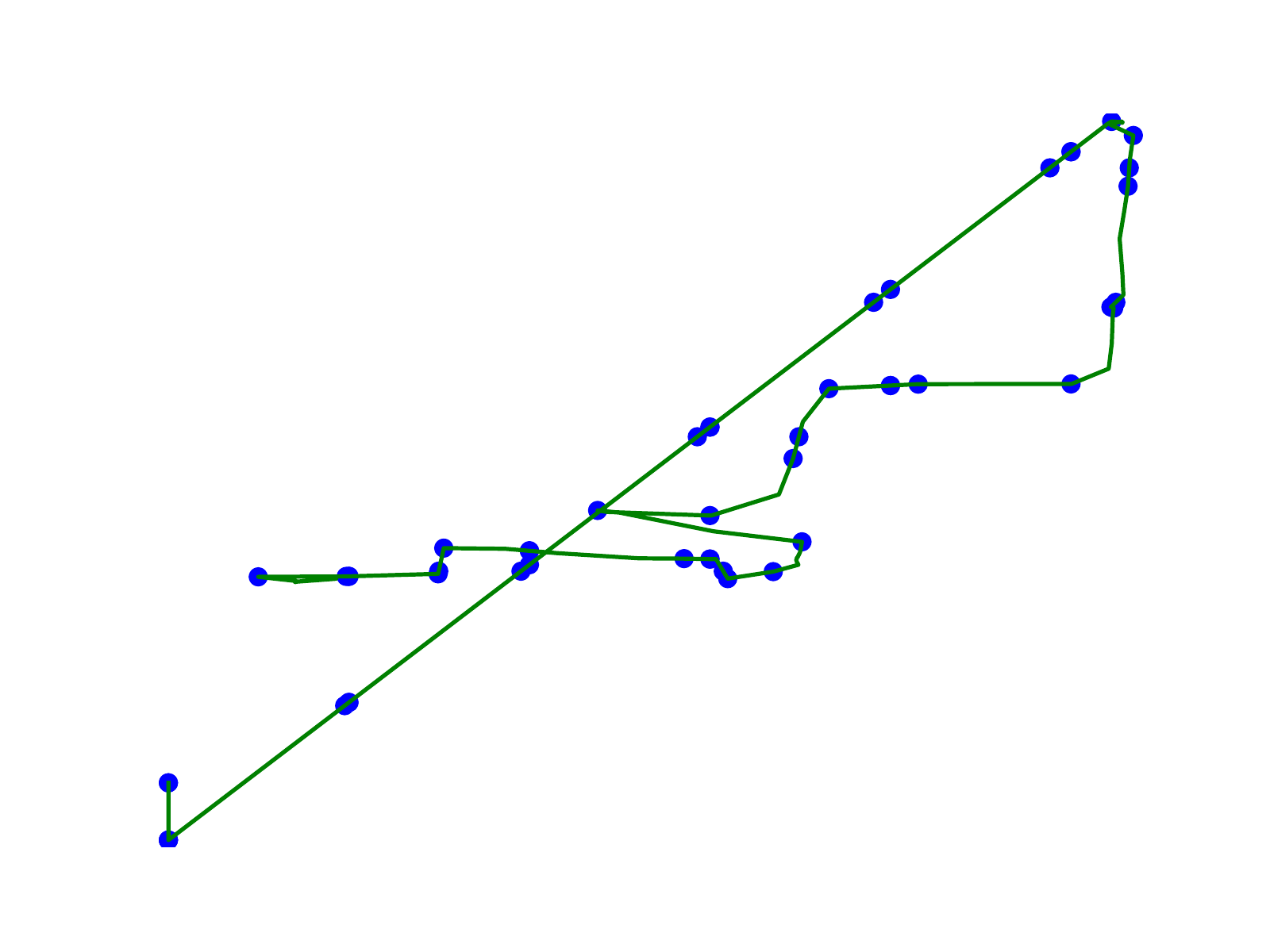}
	\includegraphics[width=0.19\linewidth]{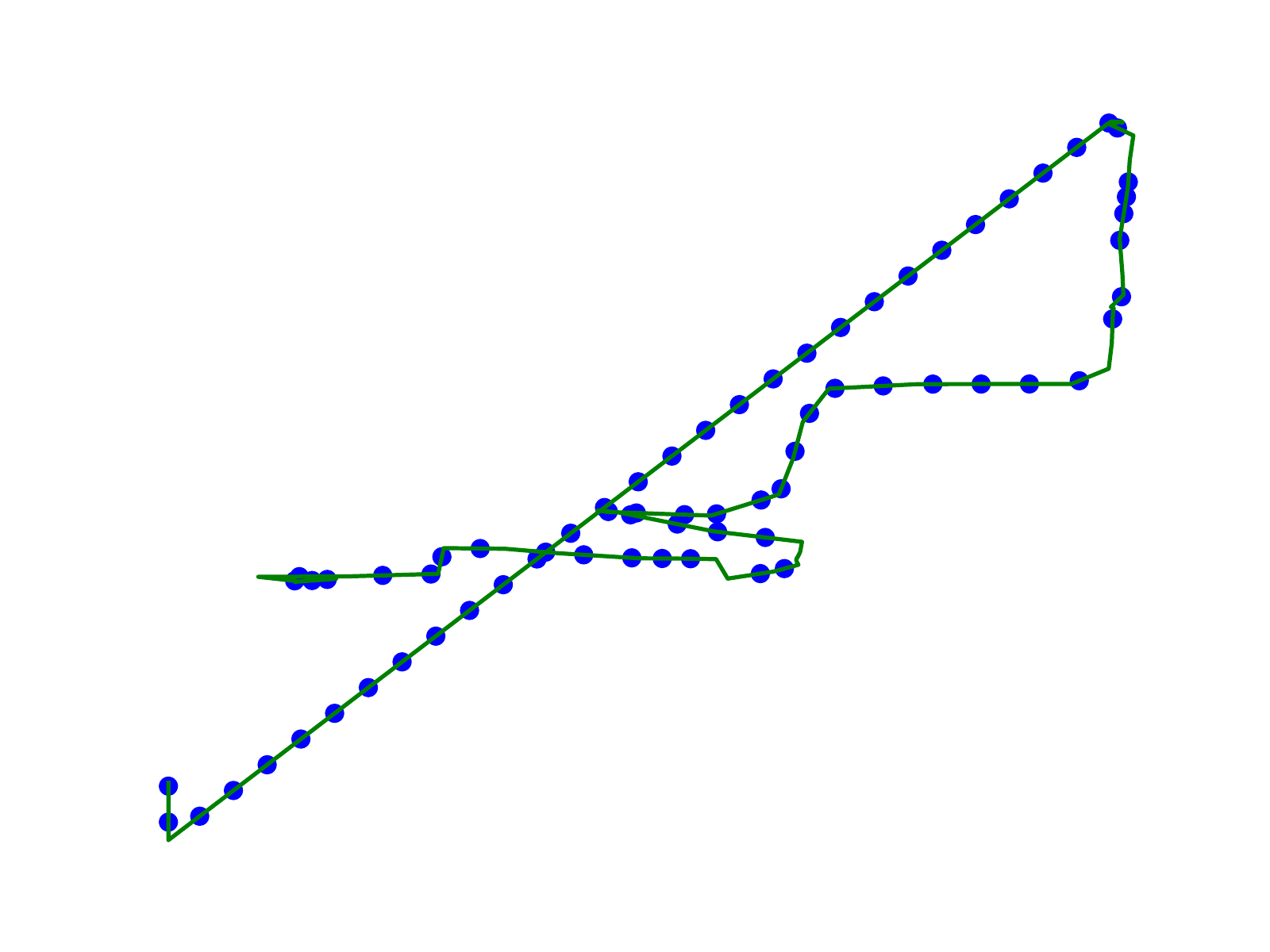}
	
	\caption{Full scanning coresets from left to right:  \Fmdp for halfplanes, \Fmk for halfspaces, \Fmgd for disks (radius $r$), \Fmgd for disks (radius $2r$), and \Fme.}
	\label{fig:blockpartial}
\end{figure*}
We next list a series of trajectory approximations we study; all bounds assume all trajectories lie in a $[0,1] \times [0,1]$ domain, otherwise $\alpha$ is scaled accordingly.

\myParagraph{All Waypoints.}
This baseline simply sets $P'_t = P_t$ and retains all waypoints.  
This does not deal with long segments well, and does not achieve an $\alpha$-approximation except for halfspaces, but may be appropriate for data collected at regular and short intervals.  
%
%

\myParagraph{Random Sampling.}
This baseline randomly samples $k$ points from segments proportional to arclength.  Let $L$ be the total arclength of a trajectory $t$.  
Based on VC-dimensional~\cite{VC71} and $\eps$-net~\cite{HW87} arguments, if $(\mathbb{R}^2,\Eu{C})$ has constant VC-dimension (as with disks, halfspaces, and rectangles) then for $k = O((L/\alpha) \log (L/\alpha))$, with constant probability it is an $\alpha$-approximation.  


\myParagraph{Even.}
In this sketch, we select $k = L/\alpha$ points evenly spaced according to arclength, where again $L$ is the total arclength of a trajectory.  This deterministically creates an $\alpha$-approximation.  To preserve the proportionality property for the partial case, we treat the trajectories as if they are chained together to adjust the first selected point from each trajectory -- so the first points on trajectory $t_j$ is a distance $\alpha$ from last point on trajectory $t_{j-1}$.

\myParagraph{DP algorithm.}
The Douglas-Peucker (DP) algorithm~\cite{HS92} is frequently used in practice as a compression step for trajectory simplification. This method iteratively removes waypoints from the original trajectory in a greedy fashion until removing another one would cause the Hausdorff distance between the original and the simplified one to exceed a chosen parameter  $\alpha$. 
This ensures, for instance, that no query shape $C \in \Eu{C}$ can intersect the original trajectory $t$ at a depth $\alpha$ into $C$ without also intersecting the simplified trajectory.  For halfspaces this provides an $\alpha$-approximation (see Lemma \ref{lem:hs-ch}).  However for rectangles and disks, this guarantee is only over the trajectory's segments, but not the waypoints $P'_t$, so does not provide an $\alpha$-approximation as desired.  

\myParagraph{Convex Hull.}
This puts all vertices on the convex hull of $P_t$ in $P'_t$.  This is a $0$-approximation (has no error) for halfspaces.  

\begin{lemma}
\label{lem:hs-ch}
A halfspace $h \in \Eu{C}$ intersects a trajectory if and only if it intersects at least one of its  waypoints.  
\end{lemma}
\begin{proof}
A halfspace $h$ intersects part (but not all) of a trajectory if and only if its boundary intersects one (or more) of its segments.  If it intersects all of the trajectory, it must contain all waypoints.  For a boundary plane to intersect a segment $s_j$, it must be that one of its waypoints $p_j,p_{j+1}$ is inside the halfspace and the other is not since the segment is a convex object with these points as the only extreme points.  Hence, we can check intersection of a trajectory $t$ with $h$ by checking if any of its waypoints are in $h$; otherwise all of the waypoints and the entire trajectory must be outside of $h$.  
\end{proof}

\myParagraph{Approx Hull.}
We create $P'_t$ as the $\alpha$-kernel coreset which approximates the convex hull of $P_t$~\cite{AHV04}.  This provides an $\alpha$-approximation for halfspaces with only $O(1/\sqrt{\alpha})$ points, independent of arclength, but with restriction that all trajectories are in $[0,1]\times[0,1]$.

\myParagraph{Lifting and Convex Hull for Disks.} 
For disks, there is a reduction via a data transformation (the Veronese Map $v$) that provides similar approximations as the convex hull approach for halfspaces.  Given a point set $P \in \R^2$, the intersection of that point set with a disk is preserved under a map to $\R^3$ where disks are mapped to halfspaces.  For $p = (p_x,p_y) \in \R^2$ we replace it with $v(p) = (p_x, p_y, p_x^2 + p_y^2) \in \R^3$.   Every disk becomes a halfspace in $\R^3$ and contains the same subset of points as the disk did in $\R^2$.  

After this transformation, set $P_t'$ as the points on the convex hull (or in the $\alpha$-kernel coreset~\cite{AHV04}) in $\R^3$.  However, because Lemma \ref{lem:hs-ch} does not apply to disks, this does not have any approximation guarantee.  That is, a disk $C$ in $\R^2$ which intersects a segment but no waypoints, transforms to a halfspace $h_C \in \R^3$ which contains part of a segment (these segments are now quadratic curves, and are not straight), but still no waypoints.

\begin{lemma}
	\label{lem:trajgrid} 
	The number of cells of a regular grid with grid cells of size $\ell \times \ell$ that a polyline can enter is $O(L / \ell)$.
\end{lemma}
\begin{proof}
	We will group cells into 9 groups and analyze each separately.  Each cell is in the same group as other cells two hops over in one of the 8 directions (left, right, $45\deg$, etc...). Each cell touches 8 other cells which are not in its group, and each one is in a distinct group. 
	Now within each group, to intersect a cell and enter another one the trajectory must travel a distance of $\Omega(\ell)$, since it will have to pass the complete vertical or diagonal distance of a cell. Thus, within a group a trajectory of length L can touch at most $O(L/\ell)$ cells. And in 9 groups, the total number of cells is at most $9 O(L/\ell) = O(L/\ell)$. 
\end{proof}

\myParagraph{Grid Kernel.  }
For disks it makes sense to adjust the approximation for different radii as smaller values of $k$ are needed for very large radii disks which can potentially intersect many trajectories. By constructing multiple approximations and scanning each with different radii significant speedups in practice can be realized.

We adjust the gridding technique for disks, and specifically for a family $\Eu{C}_r$ which only considers disks of radius at least $r$.  
We consider grid cells of size $\gamma \times \gamma$ with $\gamma = \sqrt{2\alpha r - \alpha^2 / 2}$.  Within each grid cell we retain multiple points in $P'_t$, specifically those on a $(\alpha/(2\sqrt{2} \gamma))$-kernel coreset of the points in that cell.

\begin{lemma}\label{lem:gk-err}
For a trajectory with arclength $L$, at most $O(L/(r^{1/4} \alpha^{3/4}))$ points are put in $P'_t$ for \Fmgd, and it is an $\alpha$-approximation for $\Eu{C}_r$.  
\end{lemma}
\begin{proof}
The maximum distance between two points in a grid cell is $\sqrt{2} \gamma$.  Thus the $(\alpha / (2\sqrt{2}\gamma))$-kernel incurs at most $\sqrt{2} \gamma \cdot (\alpha / 2\sqrt{2} \gamma) = \alpha / 2$ error between the convex hull of all points in that cell, and the hull of the approximate ones.  

However, a disk may intersect part of a trajectory without intersecting any of the waypoints on the convex hull.  But, if the longest possible edge in convex hull is $\sqrt{2}\gamma$ then a disk of radius $r$ not containing a waypoint can protrude into the hull at most a distance (see Figure \ref{fig:archeight}(Right))
\begin{align*}
r - \sqrt{r^2 + (\sqrt{2}\gamma/2)^2} 
&=
r - \sqrt{r^2 + (\sqrt{2 \alpha r - \alpha^2 /2}/\sqrt{2})^2}
\\&  =
  r - \sqrt{r^2 + (\alpha r - \alpha^2/4)}
\\&  = 
  r - \sqrt{(r + \alpha/2)^2}
  =
  r - r + \alpha/2
  =
  \alpha/2.
\end{align*}

The sum of these two errors is at most $\alpha/2 + \alpha/2 = \alpha$, as desired.  
  
The total number of points associated with a trajectory of length $L$ is:
the number of cells it intersects $O(L/\gamma)$ 
times
the number of points in each cell $O(1/\sqrt{(\alpha / (2 \sqrt{2} \gamma}) = O(\sqrt{\gamma/\alpha})$.  
In total this is 
\[
O\hspace{-0.5mm}\Big(\hspace{-0.3mm}\frac{L}{\gamma}\hspace{-0.3mm}\cdot \hspace{-0.3mm}\sqrt{\gamma/\alpha}\hspace{-0.3mm}\Big) 
=\hspace{-0.5mm} 
O\hspace{-0.5mm}\Big(\hspace{-0.3mm}\frac{L}{\sqrt{\alpha \gamma}}\hspace{-0.3mm}\Big)
=\hspace{-0.5mm}
O\hspace{-0.5mm}\Big(\hspace{-0.3mm}\frac{L}{\sqrt{\alpha \sqrt{\alpha r}}}\hspace{-0.3mm}\Big)
=\hspace{-0.5mm}
O\hspace{-0.5mm}\Big(\hspace{-0.3mm}\frac{L}{\alpha^{3/4} r^{1/4}}\hspace{-0.3mm}\Big)\hspace{-0.3mm}.
\]
\end{proof}

\section{Trajectory Sampling}
\label{sec:VC}
To enable efficient scanning for the full model we require that the number of regions grows polynomially with the number of trajectories as otherwise random sampling cannot be used to attain an additive approximation bound.
For the trajectory range spaces we consider it was not previously known if this was true. A bound on the VC-dimension of these range spaces would ensure this.  
It turns out the VC-dimension bounds are tied in some manner to $k$, the number of points representing each trajectory; the number of possible subsets is then a function of the range complexity ($\nu$, e.g., dimension $d$) and the trajectory complexity ($k$). 

We now consider that each trajectory $t$ is represented by exactly  $k$ labeled points $P_t'$ (if it is less than $k$, we can duplicate the last point to increase to $k$).  That is $P'_t \subset \mathbb{R}^{d \times k}$.  
Next consider an alternative range space $(\mathbb{R}^{d \times k}, \RangeSet_{k})$ where the ground set is the approximate waypoints. 
The number of ranges induced on a set of labeled points (two ranges are the same if they contain the same set of labels) is upper bounded by the number of unique subsets on the $mk$ unlabeled points. Therefore the growth function on $m$ sets is upper bounded by $O((mk)^{\nu}) = m^{O(\nu \log k)}$, where the base range space $(\mathbb{R}^{d}, \RangeSet)$ has VC-dimension $\nu$. 

\begin{lemma}
	\label{lem:alphasample}
	The growth constant for $(\mathbb{R}^{d \times k}, \RangeSet_{k})$ is $O(\nu \log k)$ where $\nu$ is the growth constant of the range space $(\mathbb{R}^{d}, \RangeSet)$.  Hence the VC-dimension is $O(\nu \log k)$.  
\end{lemma} 

In our context, this means after trajectories in $\mathbb{R}^2$ are represented by at most $k$ points, then the VC-dimension for ranges defined by disks, halfspaces, or rectangles have VC-dimension $O(\log k)$.  

So as long as the number of trajectory waypoints is bounded by $k$, the sample sizes for $S$ and $N$ (needed in the two-level framework~\cite{SSSS,MP18b}) are increased by a rather benign near-logarithmic in $k$.   We will invoke this in the context of scanning algorithms in Section \ref{sec:full-scanning}.

\myParagraph{Bounded $k$ is needed.}
It seems hopeful that a better bound independent of $k$ may be possible, but we can show that for halfspaces, disks, and rectangles it is possible to construct cases where the complexity of the trajectories, and hence the VC dimension is unbounded. To see this we will first restrict to the set of ranges $\Eu{H}_d$ induced by halfspaces in $\R^d$.  Indeed, we can replace each trajectory with a convex set by Lemma \ref{lem:hs-ch}, so we only need to work with a ground set of all convex sets $\mathbb{C}$.  However, if the complexity of the trajectories, and hence the convex sets, is unbounded, then so is the VC-dimension even in $\R^2$, as the next lemma shows.  

Because disks in $\mathbb{R}^d$ are special cases of halfspaces in $\mathbb{R}^{d+1}$ (by the Veronese map), this bound holds for disks as well.

\begin{lemma}
\label{lem:VC-H}
The VC-dimension of $(\mathbb{C},\Eu{H}_2)$ is unbounded, and hence so is the VC-dimension of $(T, \Eu{H}_2)$ with no restriction on $k$.    
\end{lemma}
\begin{proof}
For any integer $z$, we can design a set of $z$ convex sets $c_1, c_2, \ldots, c_z \in \mathbb{C}$ in $\R^2$ so that we can shatter the set.  Let each $c_j$ be a nearly identical convex polygon with $2^z$ vertices.  Each of the $2^z$ vertices, which are in nearly the same location for each polygon, corresponds with one of the $2^z$ different subsets we seek to define.  If the vertex is in a polygon that should be in the subset, shift it slightly counter-clockwise; if it is not in the subset, shift it slightly clockwise.  Then a halfspace can include the vertex with all points shifted counter-clockwise (those intended for the subset) and nothing else in any polygon.  Since we can do this independently for all $2^z$ subsets for any $z$, we can shatter a subset of any size $z$.  

Now because a halfspace intersecting a trajectory $t$ is equivalent to intersecting any of its waypoints (see Lemma \ref{lem:hs-ch}), then it is equivalent to intersecting the convex hull of those waypoints.  Thus, we can generate the same construction with a trajectories with $k=2^z$ waypoints for any value $z$, and thus if we have no bound on $k$, we have no bound on the VC-dimension of $(T,\Eu{H}_2)$.  
\end{proof}

For rectangles we can construct a similiar proof.

\begin{lemma}
\label{lem:VC-R}
The VC-dimension of $(T,\Eu{R}_2)$ is unbounded if there is no restriction on $k$.    
\end{lemma}
\begin{proof}
Again for any integer $z$, we can design a set of $z$ trajectories $t_1, t_2, \ldots, t_z \in T$ in $\R^2$ so that we can shatter the set. Each $t_i$ will have $z!$ vertices such that for each trajectory $t_i$ it has 
$t!/2$ vertices where each vertex takes part in a permutation of $t_1, t_2, \ldots, t_z$ on the line $y = x$ and similarly $t!/2$ vertices on the line $y = x + 1$ partaking in their own set of permutations. Between the two lines every permutation can be constructed. The corner of a rectangle can then cut off a subset of each permutation independendently of intersecting another subset and since all possible $2^z$ subsets are contained in the permutations, every subset of trajectories can be induced.  Thus if we have no bound on $k$, we have no bound on the VC-dimension of $(T,\Eu{R}_2)$.  
\end{proof}

\section{Scalable Algorithms for Finding Trajectory Anomalies}
\label{sec:algo}

We next describe how to efficiently scan over the trajectory range spaces to efficiently find $\eps$-approximate spatial scan anomalies on the various range spaces defined for trajectories, and statistical discrepancy functions $\Phi$.  In the case of the flux and partial models, we provide new direct reductions to the point-set based scanning algorithms.  For the full model these reductions are not possible, and we require the development of several new insights -- different ones for each scanning shape.  In particular, for scanning under the full model with disks (perhaps the most intuitive definition) we develop new ways (the MultiScale Disk approach) to represent and approximate the range space which becomes much more efficient that what was even previously known about point-set based scanning.

\subsection{Reductions for Partial and Flux Models}
\label{sec:reduction}

We first describe two reductions for the flux model and partial range model to algorithms for scanning over points instead of trajectories.  

\myParagraph{Flux model reduction.}
For the flux model, the reduction starts again by sampling trajectory subsets $S,N \subset T$.  Now we convert each trajectory $t \in S$ (or in $N$) to a point set in $S_P$ (or $N_P$) as follows.  We first convert every trajectory $t$ into only its two endpoints $p_1$ and $p_k$ and place both of these points in $S_P$ (or in $N_P$).  
In $S$ we require $r$ (recorded) and $b$ (baseline) values, and we only focused on the simpler variant which considers the linear model $\phi(r(C),b(C)) = |r(C)-b(C)|$.  
The we set
$b(p_k) = b(t)$, $r(p_k) = -r(t)$ and $b(p_1) = -b(t)$, $r(p_1) = r(t)$.  
Note now that if both $p_1,p_k \in C$, then the total contribution $r(C)$ and $b(C)$ is $0$; the points cancel each other out.  When only $p_1 \in C$, then the contribution is $r(p_1) - b(p_1) = r(t) - b(t)$ as desired, and if only $p_k \in C$, then the contribution is $r(p_k) - b(p_k) = -r(r) + b(t) = - (r(t) - b(t))$, also as desired.  

\begin{theorem}
	A flux model scan statistic for the linear statistical discrepancy function $\Phi$ on trajectories can be reduced to a point-based scan statistic on the endpoints.  
\end{theorem}

\myParagraph{Partial model reduction.}
For the partial range model, the key quantities for $\Phi(C)$ are $r(C)$ and $b(C)$, the fraction of all possible contributions from trajectories from the recorded and baseline sets.  Since in the partial range model we restrict to parts of trajectories, independently of which trajectory they are part of, we can convert to a point set input as follows.  Given the full sets of trajectories $T$, we denote the continuous set of points in these trajectories as $X_T$. We  
then take uniform (or weighted) samples of $X_T$ to construct $S_T$ and $N_T$.

Since the contribution of a point in $S_T$ towards $\Phi(C)$ is independent of the contribution of other points on the trajectory (unlike the full model), running a point-based scanning model on $S_T,N_T$ will return the same $\Phi(C)$ value for any $C$ as the trajectory-based partial range model.    

\begin{theorem}
A partial range model scan statistic on trajectories can be reduced to a point-based scan statistic on a uniform sample over the trajectory by arclength.  
\end{theorem}

\subsection{Scanning under the Full Model}
\label{sec:full-scanning}
There are three major challenges in extending scanning algorithms to the full model -- even after first converting each trajectory $t$ into a point set $P'_t$ of size $k$.  
The resulting approaches are multi-faceted, and different for each scanning shape, and summarized in Table \ref{tab:alg}.  

First, in order to obtain the runtime bounds of point-based two-level sampling algorithms, the sets $N$ and $S$, were of size roughly $n = \nu/\eps$ and $s = \nu/(2\eps^2)$ in the point-based model, now need to be of size $n_k = n \cdot k$ and $s_k = s\cdot k$, respectively.  Each object placed into the ``net'' or ``sample'' set now is required to be a set of $k$ points (on average) from each of the trajectories sampled.  The scanning algorithms have linear time in $s_k$ and moderate polynomial (degree $1$ to $3$) in $n_k$, so this increases the runtimes beyond the point-based setting by a moderate polynomial factor in $k$.  The results for the various spatial approximation techniques are summarized in Table \ref{tab:spatial-err}.  
We will demonstrate that for halfspaces and rectangular ranges, this increase is tolerable if the right $\alpha$-spatial approximation is used, but for disks we design a new multi-level approximate scanning approach which works in tandem with the $\alpha$-spatial approximation.  

Second, extra bookkeeping is required to maintain which point sets $P'_t$ are already intersecting a shape $C$ during the scanning process--so that a trajectory is not overcounted when multiple points intersect $C$.  
For this we use a global integer counter array $S.\mathsf{counter}$, where a non-zero counter serves as an indicator that the trajectory $t \in S$ is in the shape $C$ in question.  We maintain and only update $\Phi$ when a counter toggles between zero and non-zero.  
This extra bookkeeping may seem a trivial change, but it prevents the use of some approaches, in particular for rectangle scanning.  

Third, to extend the approximation guarantee~\cite{SSSS,MP18a} based on two-level sampling to this setting, we need to bound the VC-dimension $\nu$ of the range space $(\Eu{T},\Eu{A}_\Eu{C})$.  In Section \ref{sec:VC} we showed that for halfspaces, disks, and rectangles, if there is no bound on $k$, then the VC-dimension is unbounded.  However, we also showed for all of these shapes that when each trajectory is represented by a point sets of size $k$ (the ground set is $\R^{d \times k}$), then the VC-dimension is only $O(d \log k)$ in $d$ dimensions.   That is when the objects are point sets of size at most $k$, they only increase the VC-dimension by a benign $\log k$ factor, and hence the coreset sizes for $n_k$ and $s_k$.   These sample bounds are incorporated into Table \ref{tab:comb-err}.

\begin{table}
\begin{tabular}{|p{80mm}|}
\hline
\vspace{0.2mm}
\textbf{\hspace{7mm} Runtime Overview in the Full Model} 
\\ 

\vspace{0.2mm}
SSS runtimes for the full model vary by shape, and error parameters, and methods of spatial approximation.  
From the scanning perspective they depend on $n_k = n k$ and $s_k = s k$, shown in the middle column of Table \ref{tab:alg}.   
Parameters $n$ is the small net size, and $s$ is the large sample size are described in Table \ref{tab:spatial-err} for allowing $\eps$-error in $\Phi$.  
The spatial approximation size $k$ is described in Table \ref{tab:comb-err}, and depends on the method used to achieve an $\alpha$-approximation.  
The best runtime bounds for $\eps$-error on $\Phi$ and $\alpha$-spatial error are shown in the right column of Table \ref{tab:alg}.  

\\
\renewcommand{\arraystretch}{1.5}
\begin{tabular}{|r|cc|} 
 \hline
  Shape & Runtime  & Best Bounds \\ 
\hline
 Rectangles & $O(n_k^3 s_k)$ & $\tilde{O}\left (\frac{L}{\alpha^4}  \frac{1}{\eps^2}  \log(\frac{L}{\eps \alpha}) \right )$ \\
 Halfspaces & $O(n_k s_k \log n_k)$ & $\tilde{O}\left (\frac{1}{\alpha} \frac{1}{\eps^3}  \log^3 (\frac{1}{\eps \alpha}) \right )$ \\
 Disks & $O(n_k^2 s_k \log n_k)$ & $\tilde{O}\left ( \left (\frac{L}{\alpha} \right )^3 \frac{1}{\eps^4}  \log^4 (\frac{L}{\eps \alpha}) \right )$ \\
\hline
\end{tabular}
\vspace{-3mm}
\caption{\label{tab:alg} 
Overall algorithm runtimes in terms of sample size ($n_k = n k$ and $s_k = s k$) or error parameters (statistical error $\eps$, spatial error $\alpha$, and arclength $L$).  $\tilde O(\cdot)$ hides $\log \log$ factors and is for constant probability of success.  
Derivations of bounds are in Section \ref{sec:algo}.}

\\
\renewcommand{\arraystretch}{1.2}
\begin{tabular}{|r|ccc|} 
 \hline
 Method & Size $=k$ & Error & Shapes \\ 
\hline
Convex Hull & $m$ & $0$ & $\Eu{H}_2$ \\
 DP Algorithm & $m$ & $\alpha$ & $\Eu{H}_2$ \\
 Approx Hull & $O(1/\sqrt{\alpha})$ & $\alpha$ & $\Eu{H}_2$ \\
 Grid Kernel & $O(L/(r^{1/4} \alpha^{3/4}))$ & $\alpha$ & $\Eu{C}_r$ \\
 Random Sample & $O \left ((L/\alpha) \log \frac{L}{\alpha} \right )$ & $\alpha$ &All \\ 
 Even & $O \left ( L/\alpha \right )$ & $\alpha$ &  All \\ 
 Gridding & $O(L/\alpha)$ & $\alpha$ & All \\
\hline
\end{tabular}
\vspace{-3mm}
\caption{\label{tab:spatial-err} 
The $\alpha$-spatial approximation bounds for trajectory coreset algorithms with output of size $k$.  Trajectories are of length $L$ and with $m$ waypoints, and all in a domain $[0,1]\times[0,1]$.  
Derivations of bounds in Section \ref{sec:coresets}. }

\\
\renewcommand{\arraystretch}{1.5}
\begin{tabular}{|r|c|} 
 \hline
   & Sample Size  \\ 
\hline
 Sparse Net $N$ & $n = O \left ( \frac{\log k}{\eps} \log \frac{\log k}{\eps \delta} \right ) $ \\
 Dense Sample $S$ & $s = O \left ( \frac{1}{\eps^2} (\log k +  \log \frac{1}{\delta} \right ) $ \\
\hline
\end{tabular}
\vspace{-3mm}
\caption{\label{tab:comb-err} 
Sample size bounds to obtain $|\Phi(C^*) - \Phi(\hat{C})| \le \eps$ with probability $1 - \delta$ where $C^*$ is the true maximum and $\hat{C}$ is the found approximation or trajectories of size $k$.
Derivations of bounds in Sections \ref{sec:VC} and Section \ref{sec:full-scanning}.\vspace{-3mm}}
\\ \hline
\end{tabular}
\end{table}

\myParagraph{Rectangles.}
We extend a $O(n^4 + s)$ time algorithm for scanning rectangles~\cite{MP18b}, which in our case becomes $O(\frac{1}{\alpha^3} s_k)$.  The faster algorithms from that paper~\cite{MP18b} (taking $O(n^3)$ and $O(n^2 \log \log n)$ time do not extend because they require a special decomposition for implicit processing of the ranges which cannot accommodate the maintenance of the counter.

This algorithm defines a non-uniform grid by using a scan line in the x and y direction to ensure each row and column has at most $\eps s$ trajectories or has at least $\alpha$ width whichever is larger. We recommend the \Fmg approach for an $\alpha$-approximation since then many of these coordinates are duplicated reducing the effective grid size.  Then all of the points in $S$ (of size $s_k$) are mapped into the appropriate grid cells, and duplicates for the same trajectory can be removed.  Then we consider each of the at most $n_k^4$ rectangles defined on this grid.  We can scan them efficiently by fixing every possible upper, bottom, and left side of the rectangle (there are $O(n_k^3)$ combinations), and then sliding the right side of the rectangle from the left edge until it hits the end of the grid.  The entire scanning of the right edge updates the counters at most $s_k$ times.  So the total runtime is $O(n_k^3 s_k)$. 

\myParagraph{Spatially-approximated Rectangles.}
An important optimize for rectangles takes advantage of the $\alpha$-spatial approximation.  Instead of setting consider all $n_k^3$ rectangles, we only need to consider endpoints differing by at least $1/\alpha$, there are no more than $O(1/\alpha^3)$ of these.  Thus the total runtime becomes $O(s_k / \alpha^3)$.  This variant is used in experiments.  

We can also add a restriction where we short circuit the algorithm if the resulting subgrid will have height or width greater than some set value. This restriction can significantly decrease the runtime.

\myParagraph{Halfspaces.}
To create approximate sets for each trajectory we use the \Fmh and the \Fmk methods to reduce to size $k$ point sets.  The former induces no spatial error, and the later provides and absolute bound on the size $k$ (at $k = O(1/\sqrt{\alpha})$). 

Let $n_k$ represent the total number of net points required, with $n$ trajectories sampled into $N$ and then on average requiring $k$ points $P'_t$ to approximate.
Using a combinatorial arrangement view of halfspaces and point sets, these can be scanned in $O(ns)$ time~\cite{DEM96}, using advanced techniques from computational geometry.  This can be converted into a $O(n_k s_k)$ bound in our setting.  
We review a simpler model here that takes $O(ns \log n)$ time on point sets, and $O(n_k s_k \log n_k)$ in our setting.  Let $N_k$ be the $n_k$ points approximating the $n$ trajectories in $N$, and similarly let $S_k$ be the $s_k$ points from $S$.   
For each point $q \in N_k$ sort all $s_k$ points in $S_k$ radially around $q$.  Then consider halfspaces with $q$ on the boundary, and scan through them radially around $q$ updating the the counters and $\Phi$ as necessary.  

Higher dimensional halfspaces can be reduced to lower dimensional halfplane problems by doing an affine projection down into the $2$-dimensional space. Using the halfplane algorithm to solve these problems gives a runtime of $O(n_k^{d - 1}s_k)$.  
An optimization we call the ``\textsf{Hull Trick}'' does an additional pruning step after each projection to $2$ dimensions; it creates the convex hull in $2$-dimensions, and only retains the points on the hull before scanning.

\myParagraph{Disks.}
After approximating trajectories by point sets, disk-scanning can be implemented as halfspace scanning for $d=3$ using the Veronese map. This gives runtime bounds of $O(n_k^2 s_k)$. 

However, these shapes are significantly more sensitive to the $\alpha$-spatial approximation technique, especially if there are long edges. For instance in our experimental data we found disks might need $10$ or even $50$ times more points to generate an $\alpha$-approximation for a trajectory (see Table \ref{tbl:alpha-size}). As such we design a new scanning method that works in concert with the \Fmgd approach.

\myParagraph{MultiScale Disks.}
The previous disk scanning algorithm is not tractable or scalable, but we can combine a large number of tricks to handle these issues.

We consider scanning over disks with radius in a range $[r_{\min}, r_{\max}]$, specifically where $r_{\max} = r_{\min} 2^z$ for a small integer $z$.  We decompose this into $z$ subranges, so in each the radius is in a range $[r, 2r]$ and handle these subranges independently.  In experiments, we set $z=4$ (or in some cases $z=5$) which covers a natural and intuitive set of radii.  

Then we scan in concert with the $\alpha$-spatial approximation method \Fmgd using parameter $r$.  This reduces the size of each trajectory approximation $k$, especially for large $r$; see Table \ref{tbl:alpha-size}.  
Assuming the data has been mapped to a $[0,1] \times [0,1]$ range, we can create a $1/r \times 1/r$ grid, with cell edge length $r$.  We now consider the $O(1/r^2)$ $5 \times 5$ subsets of grid cells where the center cell has at least one point in the $N_k$ set of points.  Then for each such center cell point $q$ (a ``pivot'' point), we consider only disks with $q$ on the boundary, and other boundary points among those in this $5 \times 5$ subset of $N_k$.  
Especially for small $r$, the restriction to this local subset of the data greatly reduces the number of points which are considered when constructing disks with respect to each pivot $q$: from all $n_k$ net points to only those points from $N_k$ inside this $5 \times 5$ subset. Furthermore this also significantly reduces the number of sample points in $S_k$ that must be tested for inclusion in each disk. 

Together the effects of lower bounding the radius (to reduce $k$), and upper bounding the radius (to effectively reduce the number of considered regions and the the number of points to scan $s_k$ for each pivot), makes every range $[r,2r]$ of this MultiScale Disk approach quite tractable. 

This method can be further improved with the \textsf{Hull Trick}, especially for complex trajectories.  Focusing on the disk scanning algorithm for all disks passing through the pivot $q$; this maps to the problem of scanning all $2$d halfspaces in a different lifted parameter space, as discussed in Section \ref{sec:coresets}.  This reduction allows us to apply the \textsf{Hull Trick} from the Halfspace algorithm section or another halfspace coreset method from Table \ref{tab:spatial-err} to significantly reduce the number of points that must be considered.

\section{Experiments}
\label{sec:exp}

In this section we show scalability of all of the proposed algorithms, and their effects on the statistical power of the scan statistics.  We show for the partial and flux models, through our new direct reductions to recent work, we can compute scan statistics in a scalable yet statistically powerful way.  
For the full model, the most naive reductions to existing methods are not viable, but our proposed geometric and algorithmic observations for each scanning shape lead to significant speed ups.  

We demonstrate these improvements with two types of measurements.  The first is directly showing the runtime of the algorithms as a function of either the statistical error parameter $\eps$ or the spatial error parameter $\alpha$.  
We also measure the \emph{discrepancy error}, where we attempt to find a large $\Phi(C)$ value, and we show at increasing parameter settings how the largest $\Phi(C)$ region found approaches $\arg \max_{C \in \Eu{C}} \Phi(C)$ as a function of the runtime.  
This demonstrates that these algorithms are not just fast, but they become statistically powerful in tractable runtimes.

\subsection{Setup and Data}

\label{sec:software}
All of our code and the scripts used to generate experiments are publicly available on github and our project website \footnote{https://mmath.dev/pyscan}. 
All of our code is written in C++ with an easy to use python wrapper that we hope will allow for researchers to apply these algorithms to many other data sets. Current algorithm implementations are serial, but the algorithms are embarrassingly parallel, so converting to a 
parallel implementation would be trivial.

Experiments were conducted using the python wrapper on computers with an Intel Core i7-3820 and 64GB of memory. It ran Ubuntu 14.04 with kernel version 3.13.0-147. Code was compiled with GCC version 8.1.0, Boost version: 1.69.0, and Python version 3.4.3. Experiments were run in successive fashion on a per node basis. No experiments were run 
together on the same node at the same time to minimize the effect of other processes on run time.

We perform experiments on two large trajectory data sets, described next, which have very different conditions.  One is diverse, and has trajectories of significantly different sizes and lengths, but the overlap is clustered.  The other has more uniform trajectory sizes and lengths, but they are all intermingled.

\myParagraph{Open Street Map Trace Data.}
This data set is our default, and consists of a subset of close to 6 million traces from Open Street maps with 1.282 billion total waypoints. 
We restrict the set to ones contained completely inside of a rectangular region in Europe with latitude and longitude of $[35, 60] \times [-10, 25]$.  The large extent of these traces means that there are comparatively sparse regions corresponding to rural areas and also densely packed regions such as cities.
Many trajectories are restricted to small regions compared to the full domain size. 

\myParagraph{Beijing Taxi Data.}
This is a densely packed and highly overlapping set of roughly 3 million trajectories with 129 million total waypoints collected from taxi drivers in Beijing~\cite{LJZL18}. This data set has a very high sampling frequency per trajectory with $75$\% of the points being collected less than 1 minute apart. Many roads have been driven over by hundreds to even thousands of separate traces. Since small regions can be densely packed with trajectories, local scanning approaches that restrict the region of interest to be of small size work comparatively worse. There are comparatively few sparse regions.  A set of only 25 of these trajectories are shown in the extended Full Range Model Example.  We have restricted the trajectories to be confined to a region with latitude and longitude of $[39.788, 40.094] \times [116.15, 116.612]$.

\begin{table}
	\centering
	\small
\begin{tabular}{|r|cc|}
	\hline 
	Name &  \beijingdata \#Pts & \osmdata \#Pts\\ 
	\hline 
	\Fmgd $r=.1$km &  $489$ & $117$ \\ 
	\Fmgd $r =.2$km & $286$ &  $69.6$\\ 
	\Fmgd $r =.4$km &  $191$ & $47.3$\\ 
	\Fmgd $r =.8$km &  $134$ &  $34.1$ \\ 
	\Fmgd $r = 1.6$km  &  $96.9$ &  $25.4$ \\ 
	\Fmgd $r = 3.2$km  &  $71.8$ &  $19.5$ \\ 
	\Fmdp &  $24.6$ &  $8.74$ \\
	\Fmk & $7.27$ &  $3.35$ \\
	\Fmh & $10.0$ &  $9.50$ \\
	\Fmg & $389$  & $106$\\
	\Fme & $395$ & $110.9$\\
	\hline
\end{tabular} 
\normalsize
\vspace{-2mm}
\caption{\label{tbl:alpha-size}
Setting $\alpha = 0.1km$ shows the average value of $k$ (size of approximation) for $\alpha$-approximation algorithms.}
\end{table}

\subsection{Spatial Approximation Size}
We first empirically evaluate how well the various $\alpha$-spatial approximation algorithms work on average.  We simplify all trajectories in each data set such that $\alpha$ is set to $0.1km$ -- about one city block (when normalized to $[0,1]\times[0,1]$, for Beijing this is $\alpha = 1/500$ and for OSM it is $1/\alpha = 30{,}000$).  We show the average value of $k$ for the various algorithms in Table \ref{tbl:alpha-size}.  We observe that the OSM data is much easier to approximate than the Beijing data -- and in fact we use a larger $\alpha$ as default later which would lead to even smaller $k$ values on average.  
Next we observe that the methods designed for halfspaces (\Fmh, \Fmk, and \Fmdp) generate very small average values of $k$ of 	$3$ to $10$ (or $25$ for \Fmdp on Beijing).  
For large values of $r$,  \Fmgd can start to approach $k$ of around $20$ or $30$ for OSM, but otherwise the methods for rectangles and disks (\Fmgd large $r$, \Fmg, and \Fme) have larger average $k$ of around $100$ for OSM, and over $300$ for Beijing.  
This factor of $10$ difference in $k$ can cause dramatic slow-downs in the experiments.

\begin{figure}[b]
	\includegraphics[width=0.505\linewidth]{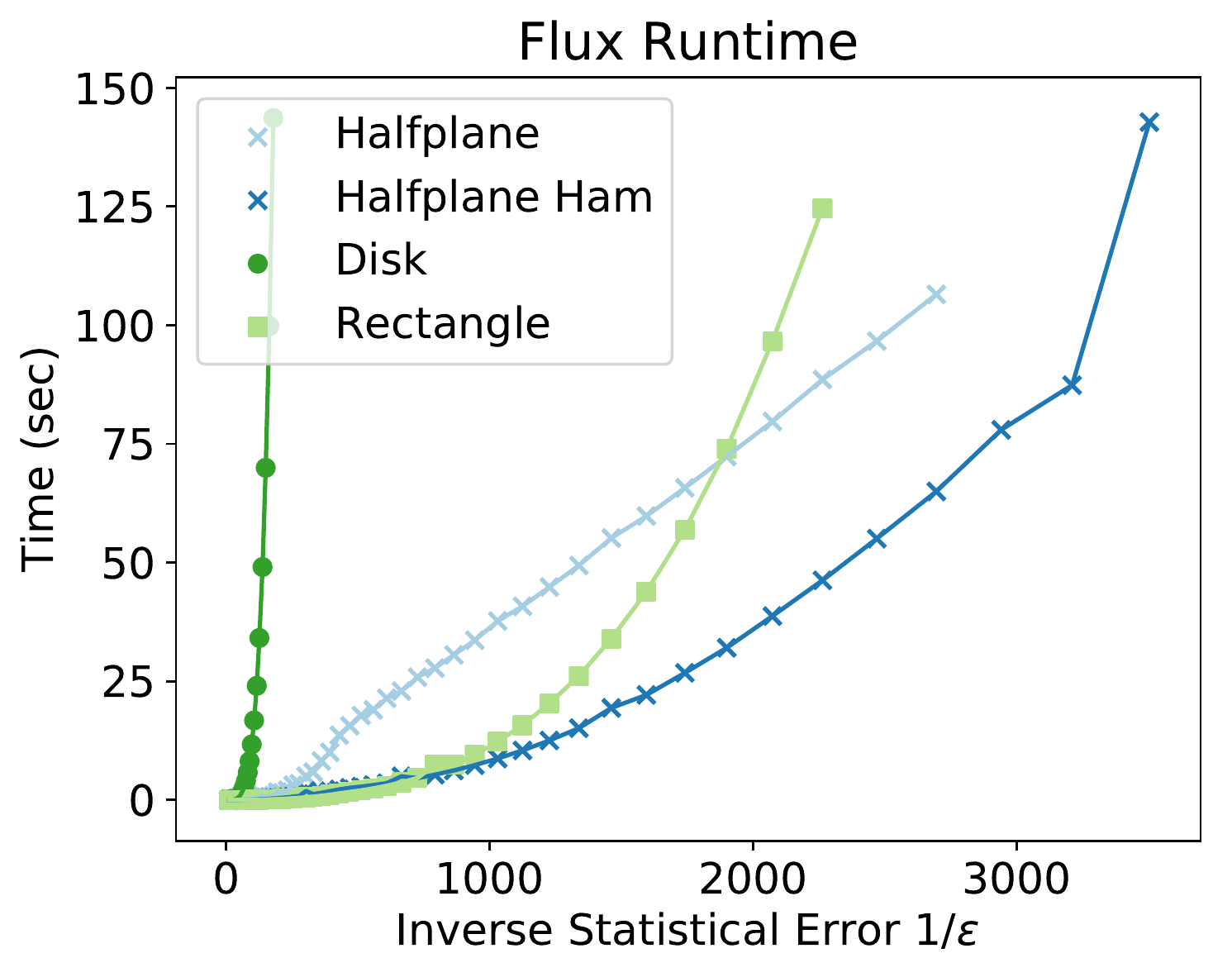}
	\hspace{-0.03\linewidth}
	\includegraphics[width=0.505\linewidth]{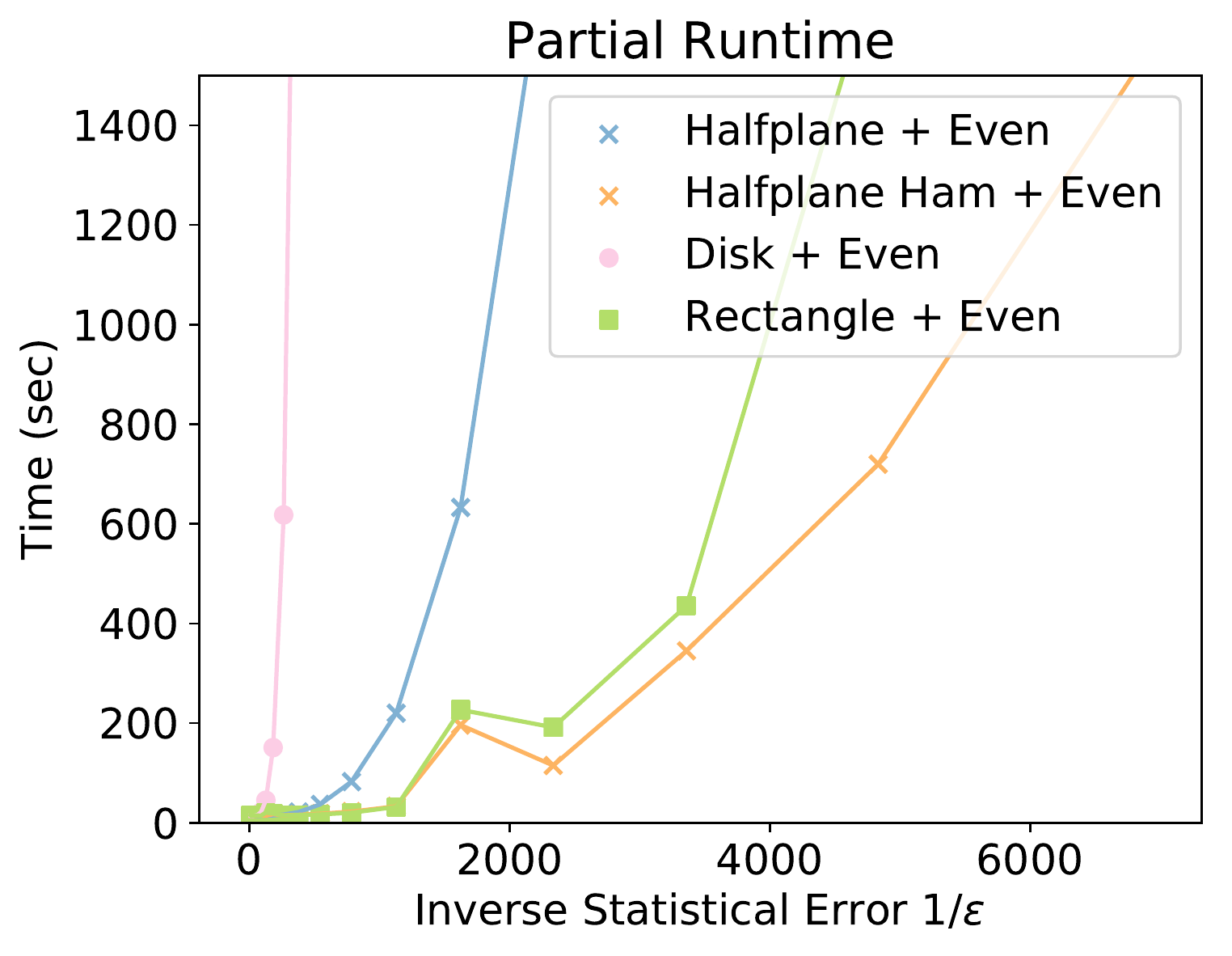}
	
	\vspace{-2mm}
	\caption{Runtime for flux model algorithms (left) and partial model algorithms (right) in terms of input net size $1/\eps$.}
	\label{fig:flux+part}
\end{figure}

\subsection{Scalability Experiments}
We next demonstrate how various parameters affect the runtime, and how various algorithms compare in scalability. 
The runtime as a function of $1/\eps$ for flux model and partial model algorithms are shown in Figure \ref{fig:flux+part}; setting $s = 1/(2\eps)^2$ and $n=1/\eps$ as suggested~\cite{SSSS}.  The partial data is sampled using the \Fme mechanism.  Some runtime curves become linear as $S$ becomes the entire data set (12 million), and only $n$ increases.

\begin{figure}
\begin{tabular}{lcc} 
&  \multicolumn{2}{c}{\textsf{Full Model Runtime on OSM}}  
\\
	\rotatebox[origin=c]{90}{Baseline Time(sec)} \hspace{-7mm} 
	&\raisebox{-0.5\height}{\includegraphics[width=0.5\linewidth]{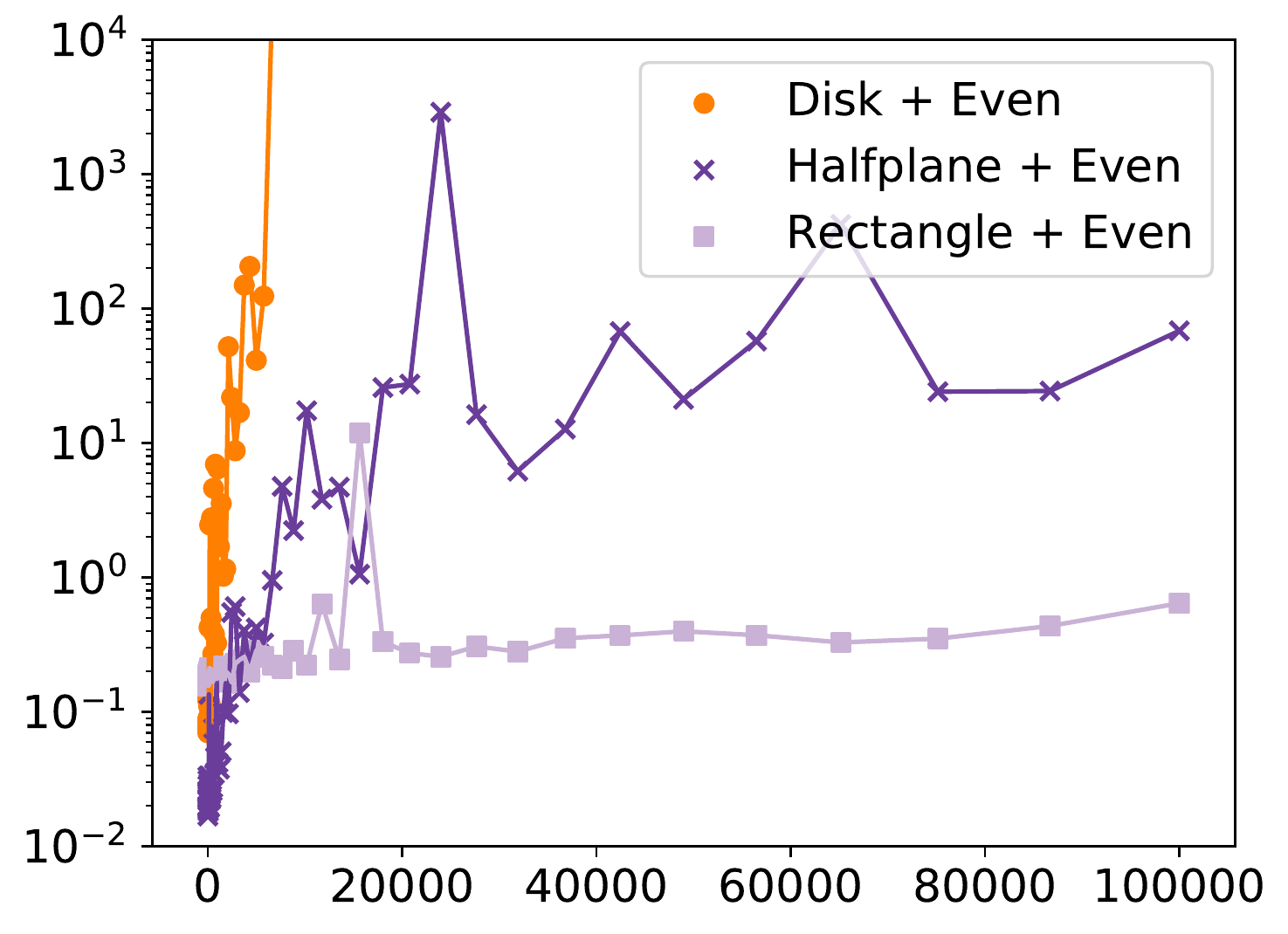}}  \hspace{-6mm} 
	& \raisebox{-0.5\height}{\includegraphics[width=0.45\linewidth]{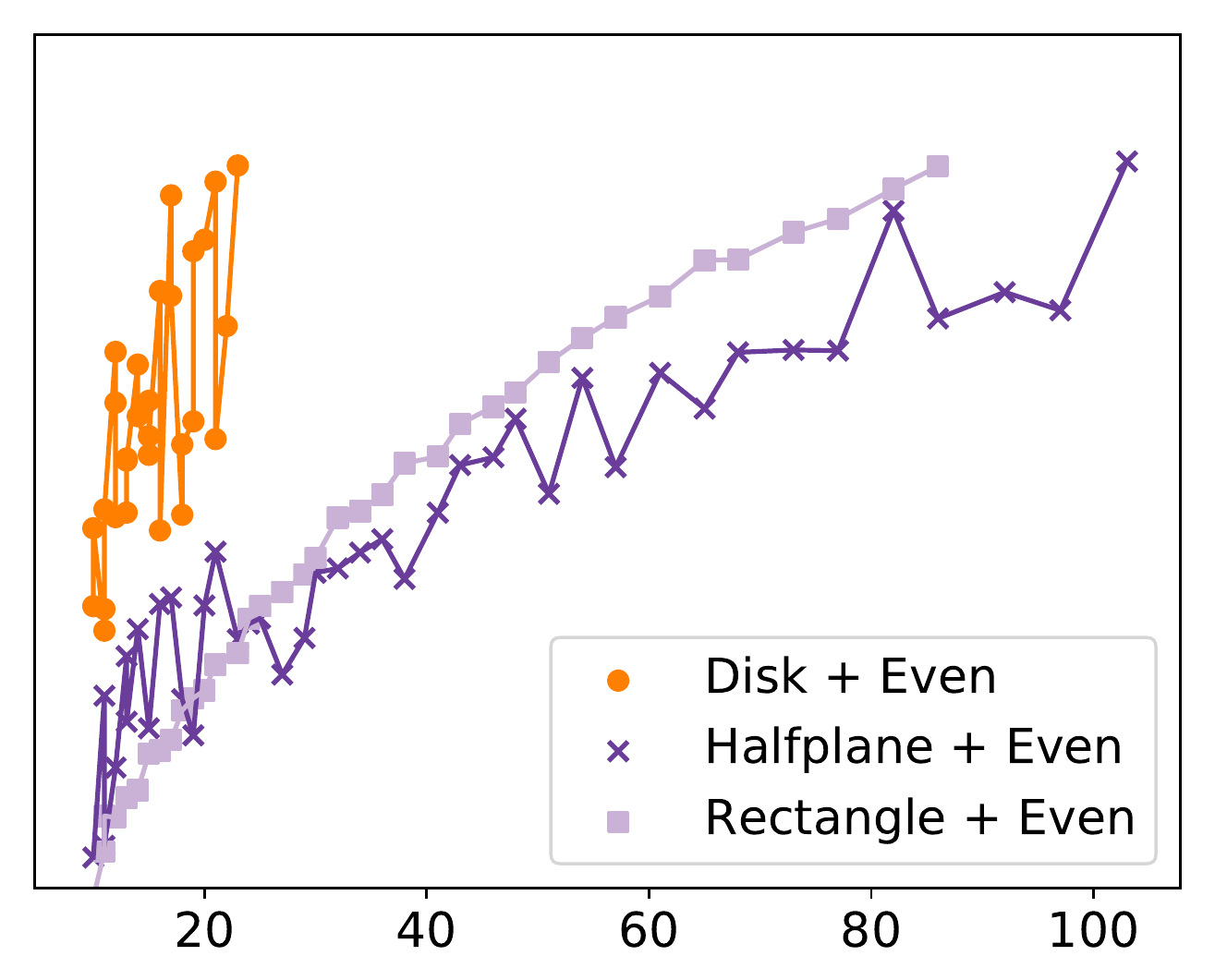}} 
	\\ 
	\rotatebox[origin=c]{90}{Halfplane Time(sec)} \hspace{-7mm} {}
	& \raisebox{-0.5\height}{\includegraphics[width=0.5\linewidth]{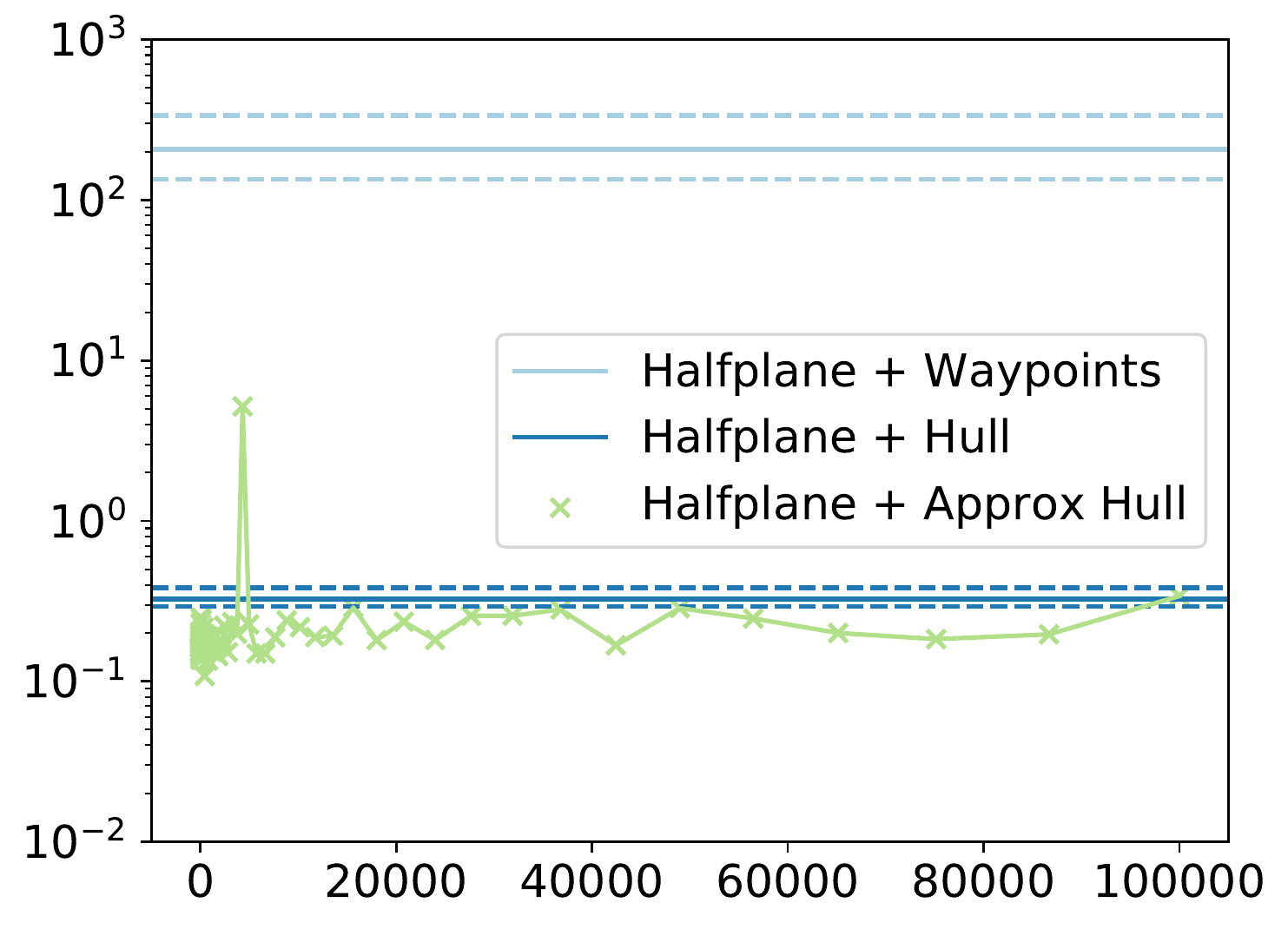}} \hspace{-6mm} 
	& \raisebox{-0.5\height}{\includegraphics[width=0.45\linewidth]{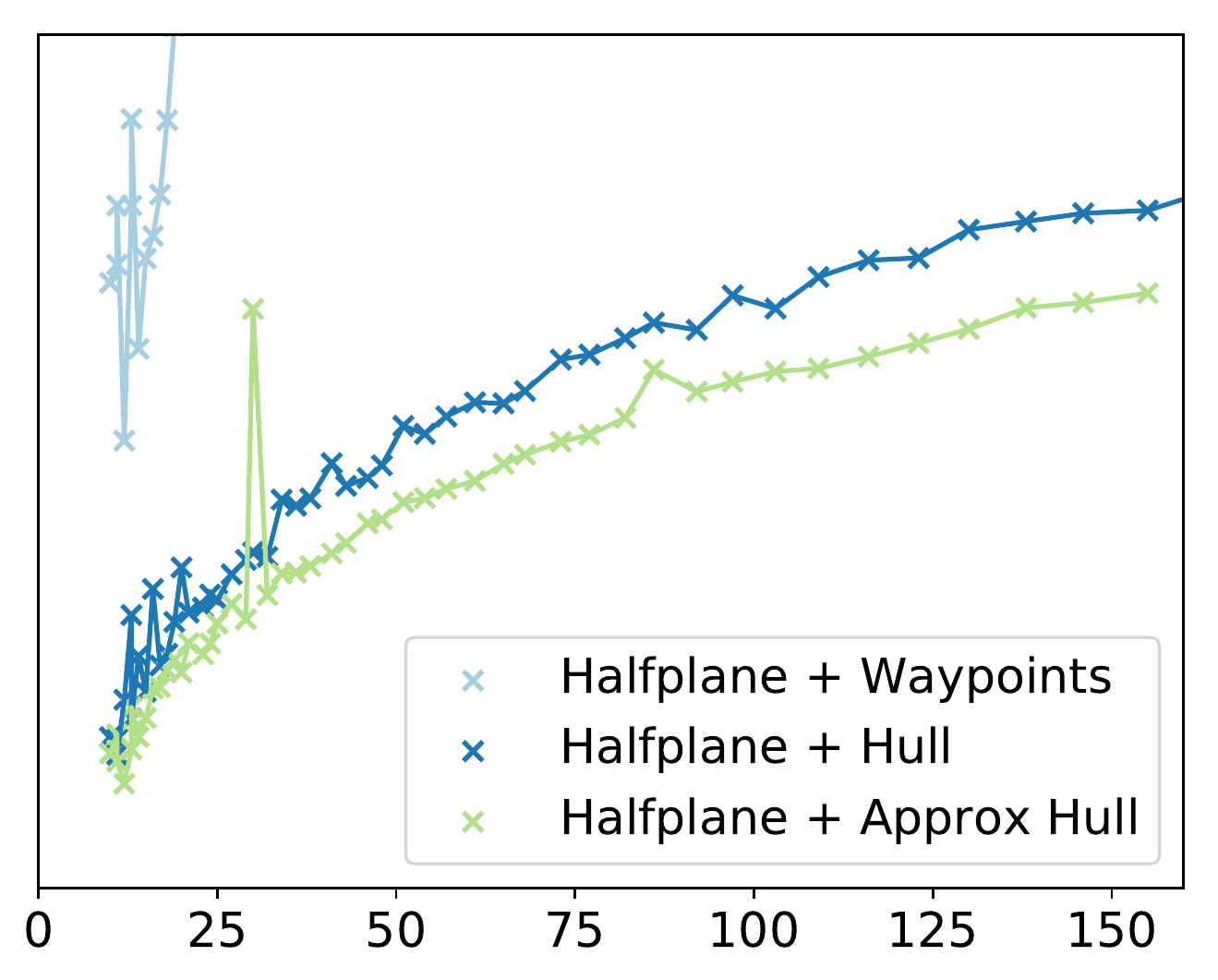}} 
	\\ 
	\rotatebox[origin=c]{90}{Disk Time(sec)} \hspace{-7mm} 
	& \raisebox{-0.5\height}{\includegraphics[width=0.49\linewidth]{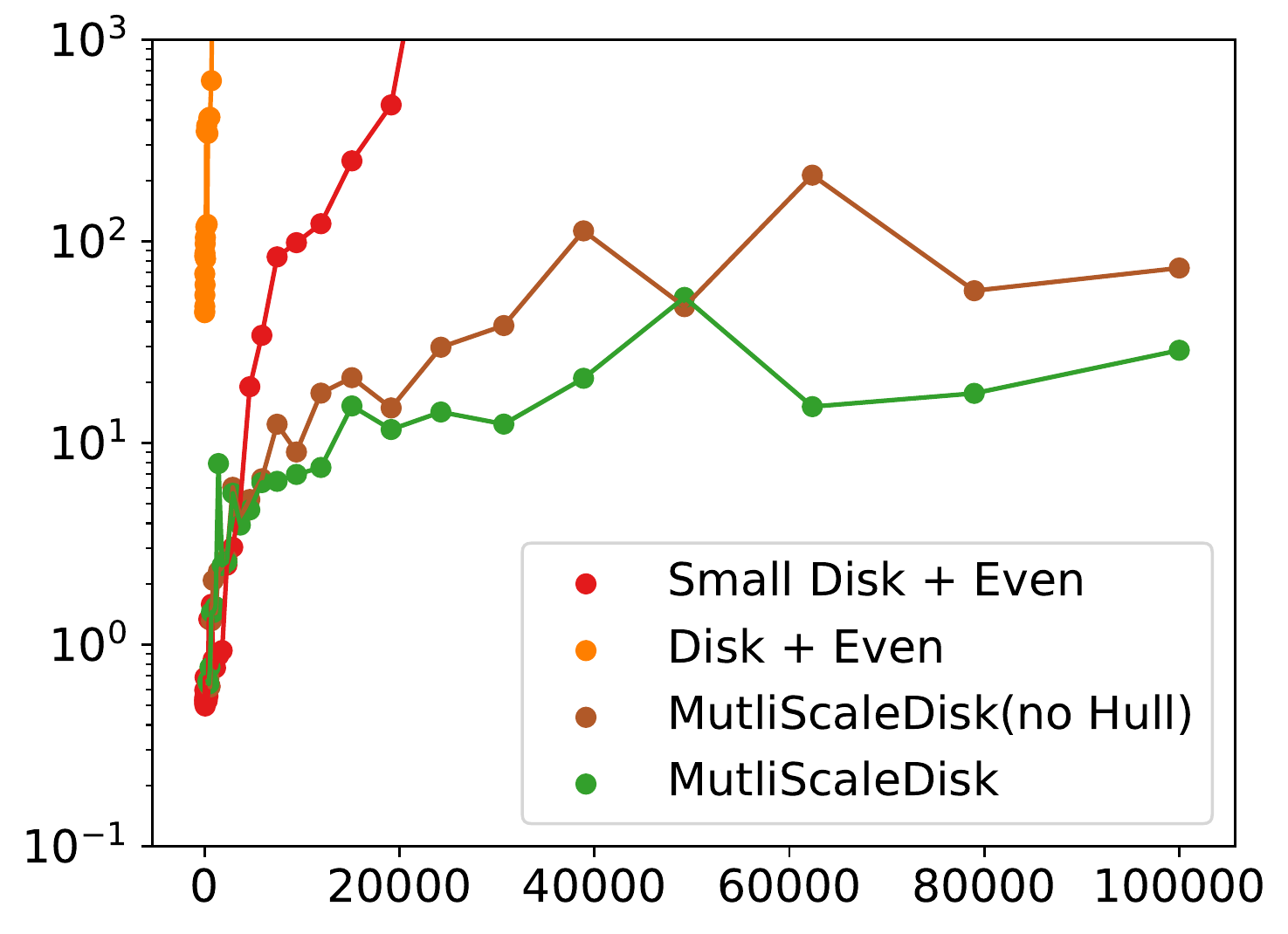}} \hspace{-4mm} 
	& 	\raisebox{-0.5\height}{\includegraphics[width=0.45\linewidth]{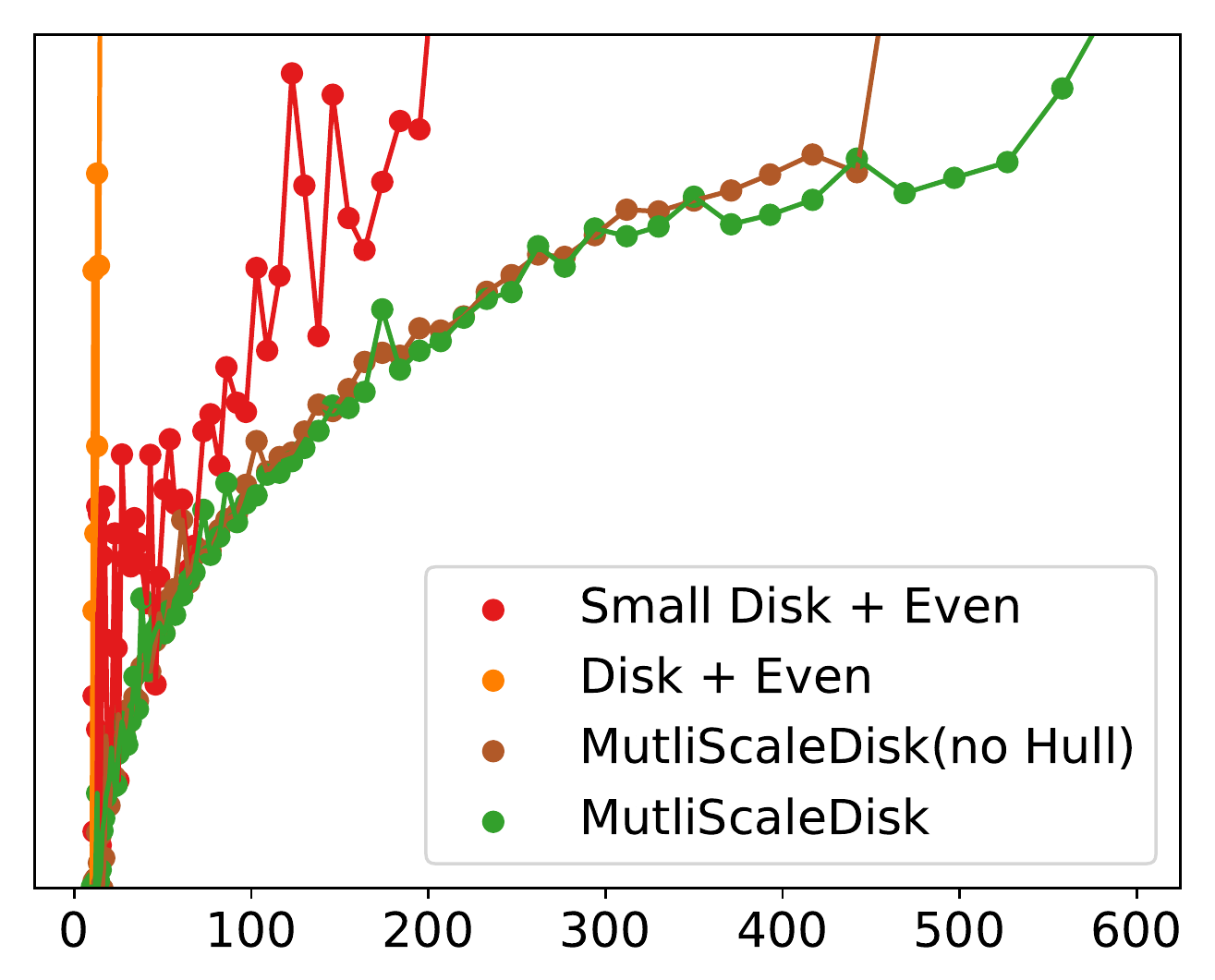}} 
	\\ 
	\rotatebox[origin=c]{90}{Best in Class Time(sec)} \hspace{-7mm} 
	& 	\raisebox{-0.5\height}{\includegraphics[width=0.49\linewidth]{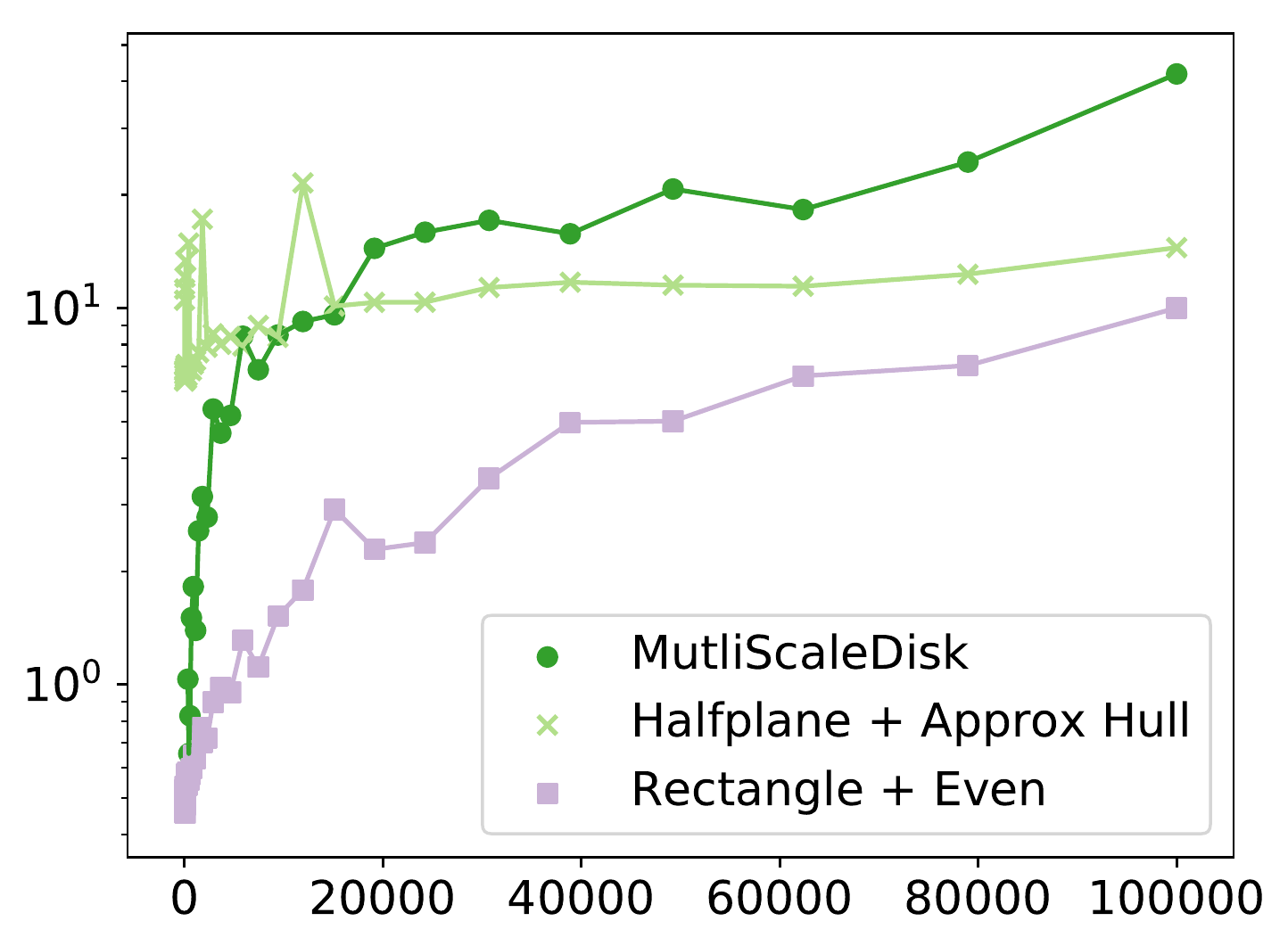}} \hspace{-5mm} 
	& 	\raisebox{-0.5\height}{\includegraphics[width=0.48\linewidth]{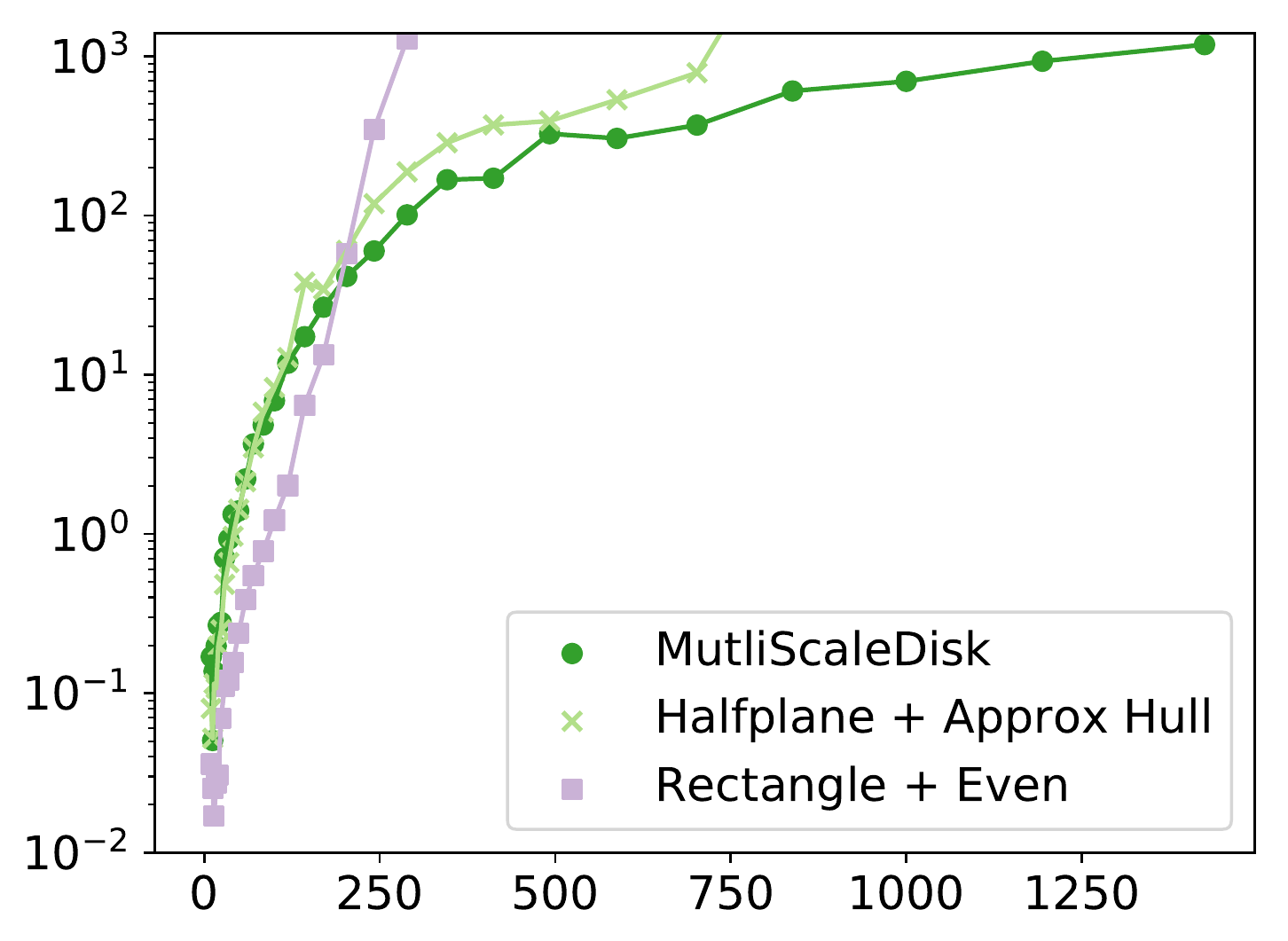}}  
	\\
		& Inverse Spatial Error $1/\alpha$ &  Inverse Statistical Error $1/\eps$ 
\end{tabular}

\caption{\label{fig:full-OSM}
	Runtime for full model scanning as function of inverse spatial error, $1/\alpha$, or inverse statistical error, $1/\eps$, on OSM data. The inverse spatial error and inverse statistical error act as a size parameter for our algorithm since our guarantees and runtime do not depend on the initial data size -- only on the statistical or spatial resolution.}
\end{figure}

\begin{figure}
\begin{tabular}{lcc} 
&  \multicolumn{2}{c}{\textsf{Full Model Runtime on Beijing}}  
\\
	\rotatebox[origin=c]{90}{Disk Time(sec)} \hspace{-7mm} 
	& \hspace{-2mm}  \raisebox{-0.5\height}{\includegraphics[width=0.52\linewidth]{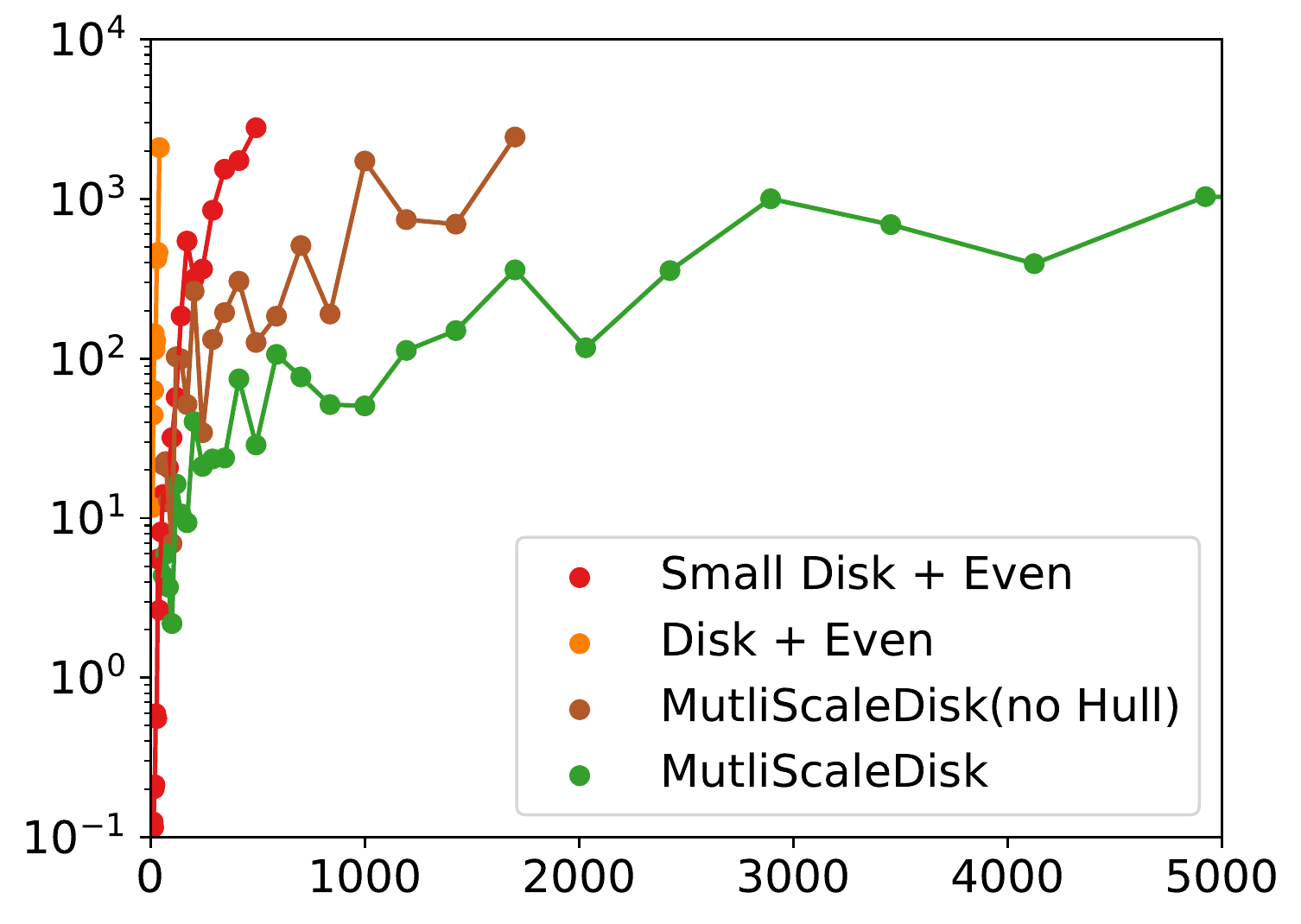}} \hspace{-4mm} 
	& 	\raisebox{-0.5\height}{\includegraphics[width=0.46\linewidth]{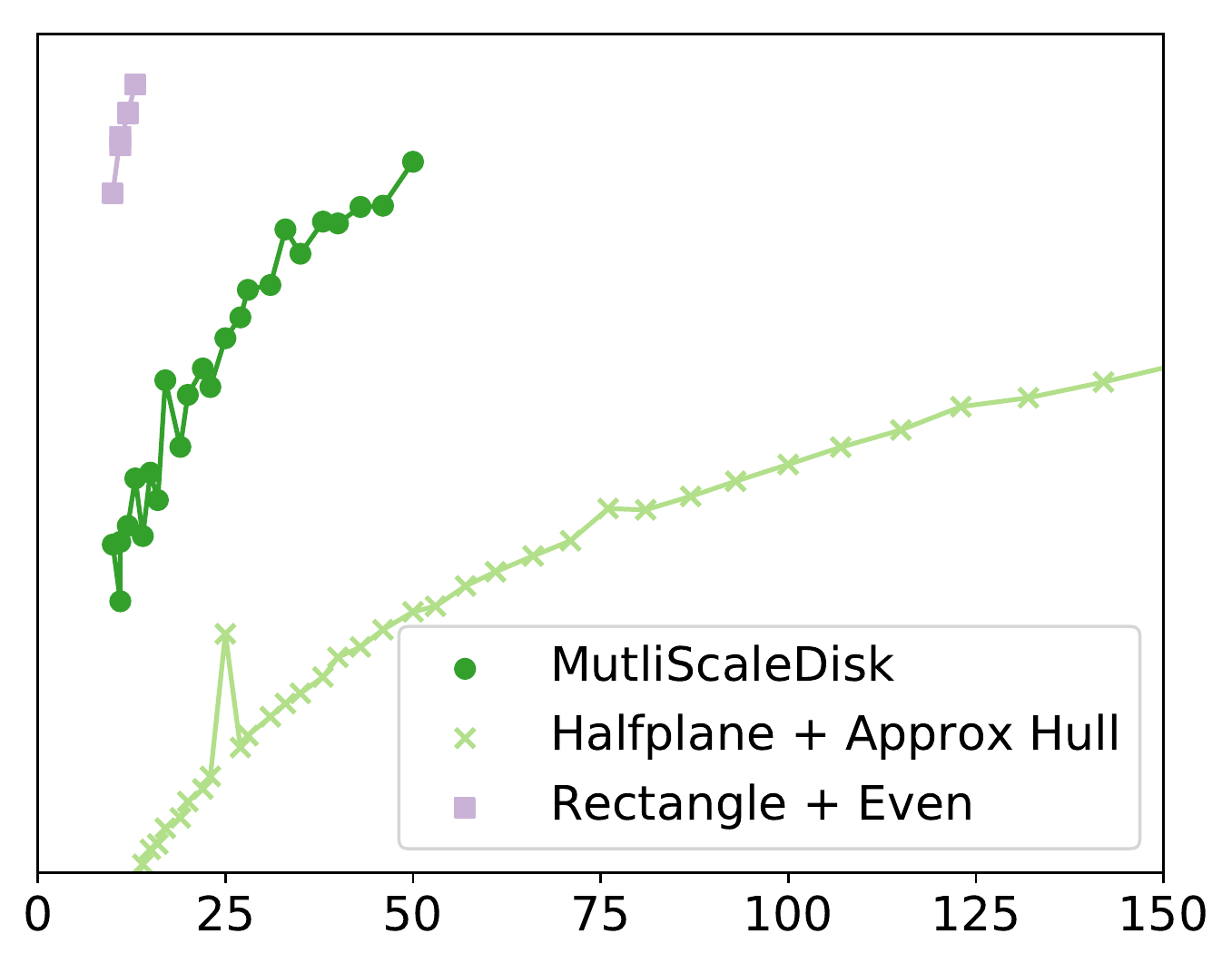}} 
	\\ 
	\rotatebox[origin=c]{90}{Best in Class Time(sec)} \hspace{-7mm} 
	& 	\hspace{-2mm} \raisebox{-0.5\height}{\includegraphics[width=0.52\linewidth]{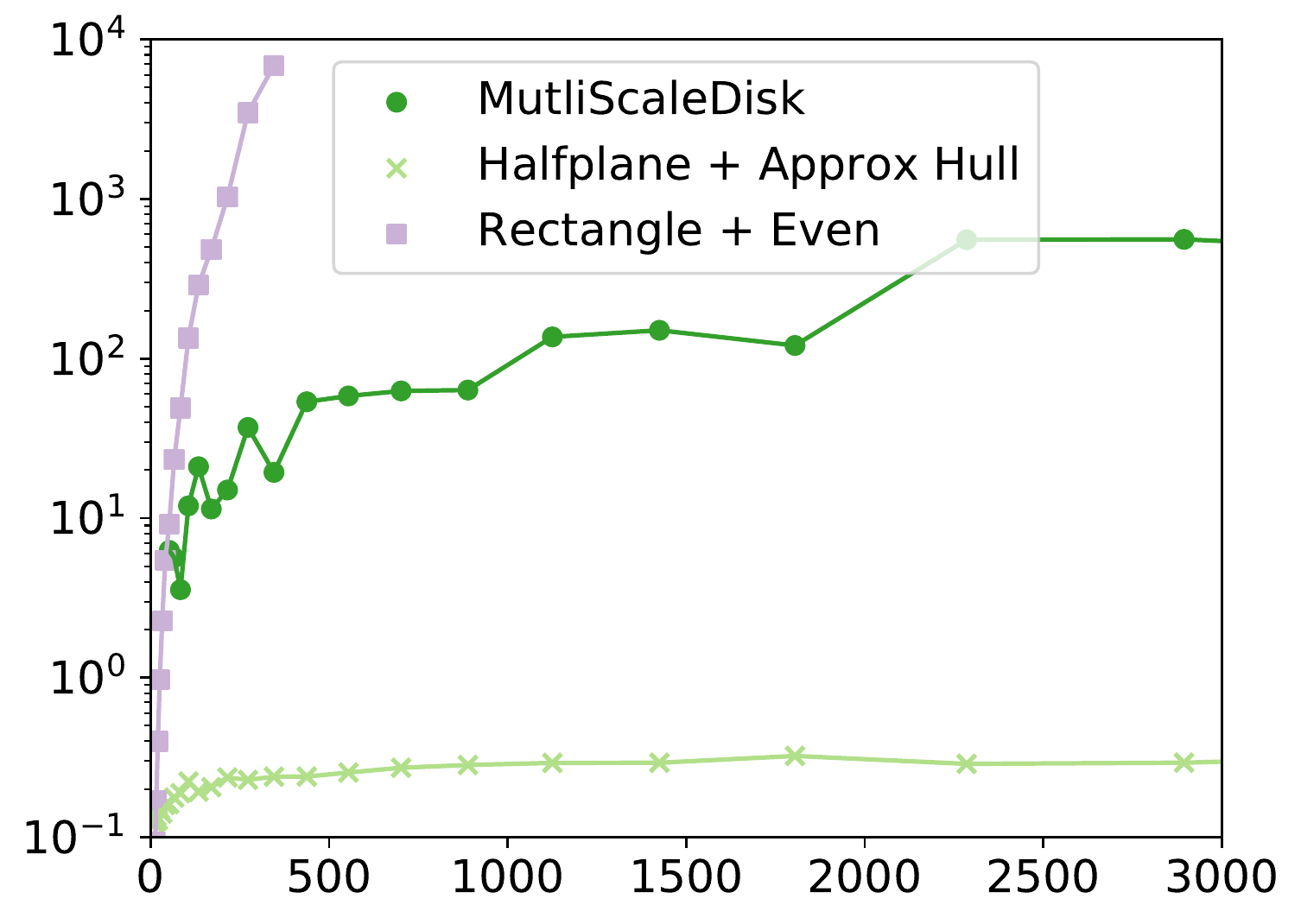}}  \hspace{-4mm} 
	& \raisebox{-0.5\height}{\includegraphics[width=0.47\linewidth]{eps_full_bc_runtime_bj}} 
	\\
		& Inverse Spatial Error $1/\alpha$ &  Inverse Statistical Error $1/\eps$ 
\end{tabular}

\caption{\label{fig:full-BJ}
	Runtime for full model scanning as function of inverse spatial error, $1/\alpha$, or inverse statistical error, $1/\eps$, on Beijing data.} 
\label{fig:bjtaxiruntime}
\end{figure}

We observe that the rectangle-based algorithms are quite scalable, and are able to set $1/\eps \approx 1{,}500$ and still complete in about $1$ minute.  The generic two-level sampling algorithm for halfspace scanning is almost as scalable, and performs better than rectangle when the ``Ham" coreset~\cite{MP18a} is used.  However, the scanning algorithms for disks require several minutes to deal with even $1/\eps \approx 100$, and becoming intractable for anything larger.  
In summary, using our new reductions to the point-based algorithms, under the partial and flux models, rectangles and halfplanes scanning for trajectory anomalies is already scalable.

\myParagraph{Full Scanning on OSM.}
The algorithms for the full intersection model are significantly more nuanced.  These runtimes are shown as a function of $1/\eps$ and $1/\alpha$ in Figure \ref{fig:full-OSM}.  While the other is varied we set $1/\alpha = 6000$ and $1/\eps = 100$ (for Disk and for Best in Class) and $1/\eps = 20$ (for Baseline and for Halfplane). Rectangles are restricted to having max side length less than $1/150$ to make them have roughly the same size as the restricted disk variants in all OSM plots.  Note that all runtimes are on a logarithmic scale.  

The Baseline rows shows the scalability of the basic algorithms using \Fme $\alpha$-spatial approximation.  Again, for halfspace and rectangle scanning, the runtimes are tolerable, but the disk scanning becomes already intractable for moderate parameter values.  
Note that now the runtimes are quite noisy due to the high variance in the trajectories length in the OSM data -- even if $k$ is small on average, for some trajectories sampled it might be quite large.  

The Halfplane row shows the difference in runtimes for the All Waypoints, \Fmh, and \Fmk methods for spatial approximation.  The All Waypoints and \Fmh have $0$ spatial error (by Lemma \ref{lem:hs-ch}), so horizontal lines with error bars are shown.  \Fmh provides a dramatic improvement over All Waypoints, of 2 to 3 orders of magnitude (i.e., from several minutes to less than a second).  The \Fmk approach shows smaller but tangible improvement over using all hull points, but adds some spatial error.

The Disk row shows another dramatic improvement in scalability as we restrict the radii considered to $r \in [1/6000, 1/300]$ and apply other improvements.  With no bound on the radius, the (Disk + \Fme) algorithms are intractable. Still invoking \Fme, but using a $1/r \times 1/r$ grid to prune the scanning to a radius range bound (Small Disk + \Fme) 
allows for moderate values of $1/\eps$ and $1/\alpha$ to complete in minutes.  
But combining the MultiScale Disk approach with the adaptive \Fmgd spatial approximation allows the same moderate error parameter runs to complete in $10$s of seconds, and in about $1$ or $2$ minutes this approach can scale to $1/\eps \approx 400$.  Adding the \textsf{Hull Trick} does not induce significant gains here.  

Finally in the Best in Class, row we show the best algorithms runtime for each scanning shape -- the improvement over the baseline for halfspaces and disks is dramatic. 
With these algorithms it is now possible to set $1/\alpha \approx 100{,}000$ and complete in about $10$ seconds for any shape.  Halfspaces and radius restricted disks can set $1/\eps \approx 500$ or $1000$ and complete in about a minute or two; they are now very scalable in $1/\eps$.  For Rectangles it can set $1/\eps \approx 250$ and still complete in a minute or two.  On the other hand, the algorithm for rectangles is more tolerant to very small values of $\alpha$.

\myParagraph{Full Scanning on Beijing.}
The runtimes for the algorithms under the full model on the Beijing data set are shown in Figure \ref{fig:full-BJ}. We restrict $r \in [1/50, 1/500]$ for radius-restricted disks, and for rectangles side length is restricted to be less than $1/25$.
The first (Disk) row shows the improvement when scanning with radius-restricted disks.  Fixing $1/\eps = 20$, only the MultiScale Disk scanning with \Fmk and the \textsf{Hull Trick} can complete $1/\alpha  > 2000$ in under an hour, and indeed can scale to $1/\alpha = 5000$ in that time.  Here the consistent long trajectories greatly benefit from the extra \textsf{Hull Trick} pruning.  Fixing $1/\alpha = 500$ only this disk scanning variant can reach a moderate $1/\eps = 50$ in under an hour.  

The second (Best in Class) row shows that now the rectangle scanning times are strictly worse than the radius restricted disk times. The halfspace algorithms have similar runtimes as with the OSM data, and are now much more scalable than the MultiDisk approach.  Thus the tangled nature of the Beijing data affects the rectangle scanning the most and the halfspace the least.


\begin{figure}
	\includegraphics[width=0.49\linewidth]{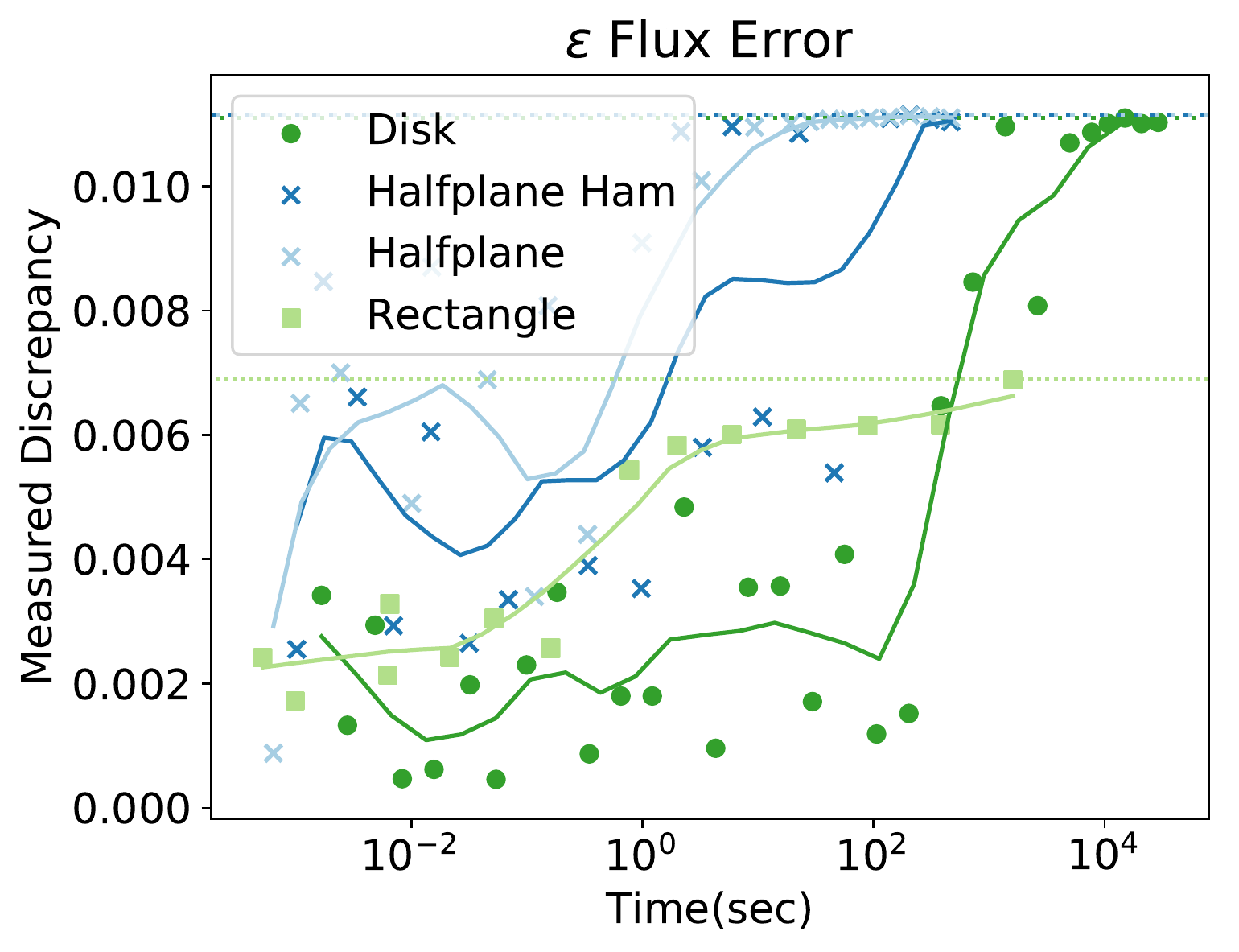} 	
	\includegraphics[width=0.49\linewidth]{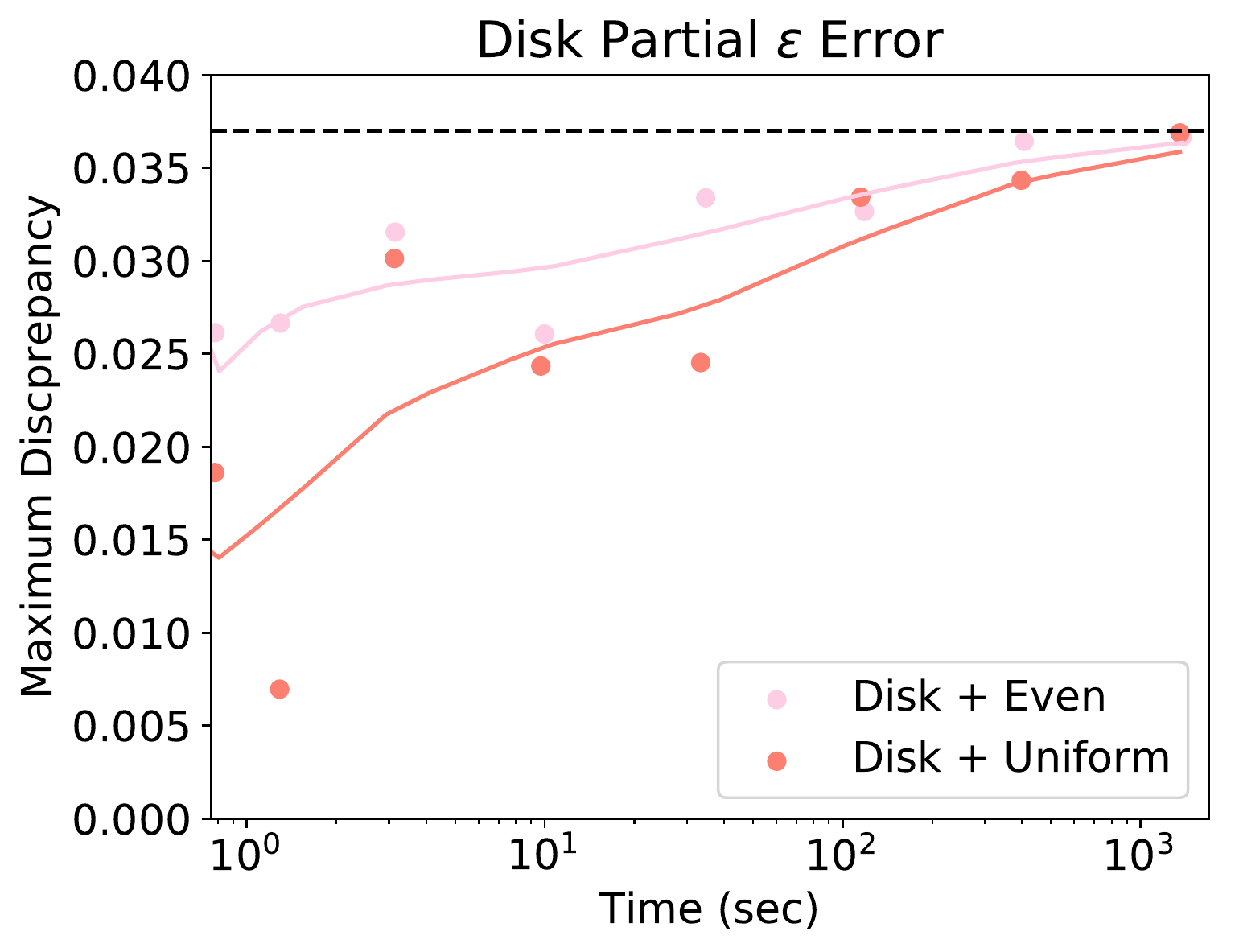}
	
	\includegraphics[width=0.49\linewidth]{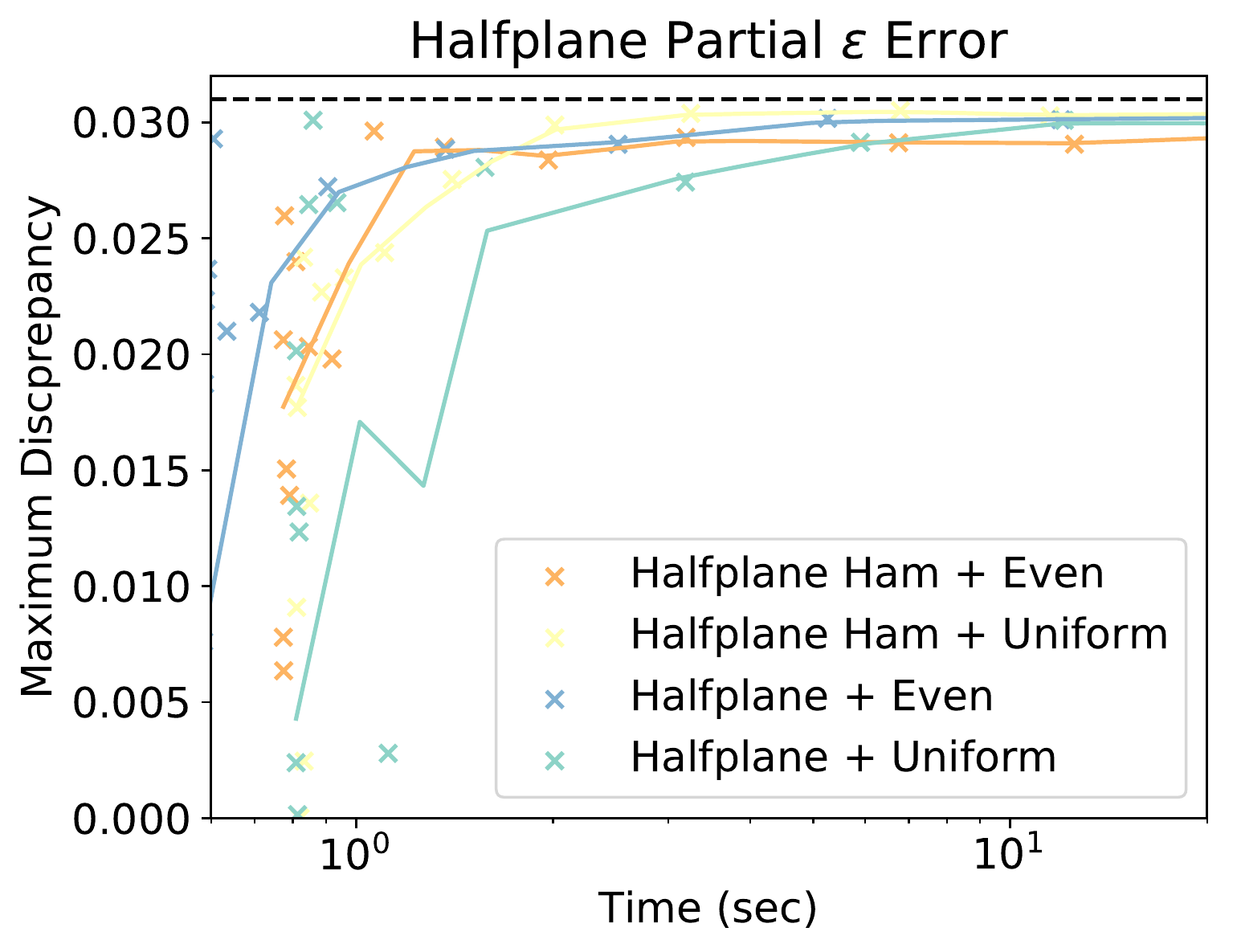}
	\includegraphics[width=0.49\linewidth]{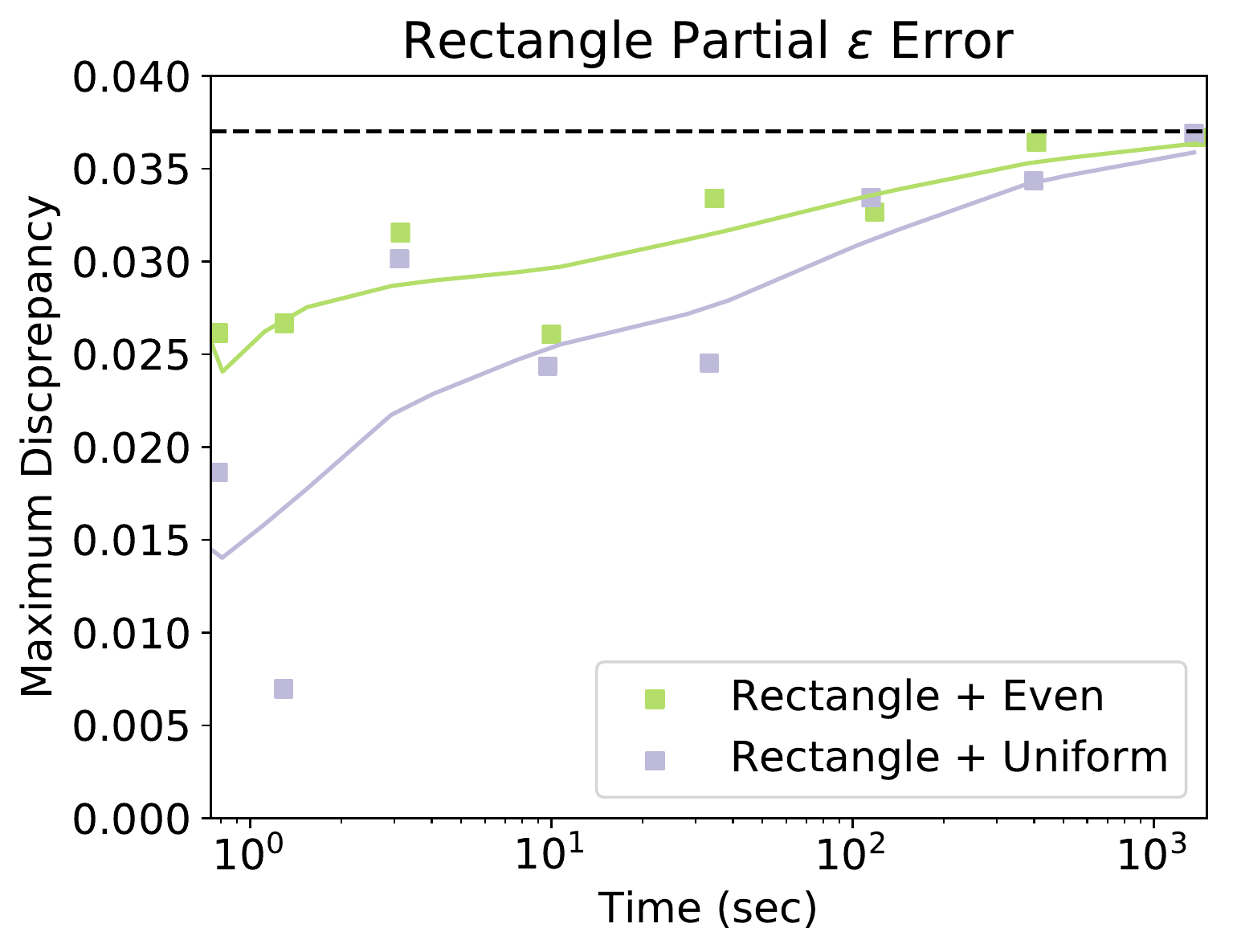}
		
	\caption{\label{fig:flux+part-err}
	Statistical power vs time for partial and flux.}
\end{figure}

\subsection{Statistical Power Tests}

In this section we measure the statistical power of these scanning algorithms.  We fix the spatial layout of trajectories from the OSM or Beijing data sets, and introduce a spatially anomalous region under the various models by adjusting the $r$ and $b$ values of trajectories.  Such regions (either a halfplane, disk, or rectangle to match the scanning region) are ``planted'' by choosing a shape $C \in \Eu{C}$ and setting a desired rates for inside $C$ as $q = r(C)/b(C)$, for all data $p = r(T \setminus C)/b(T \setminus C)$, and the anomaly size $f = b(C)/b(T)$.  The optimal shape may be different (usually slightly shifted), and so we evaluate these tests by tracking the value $\max_{C \in \Eu{C}} \Phi(C)$ (the ``Measured Discrepancy") found for various parameter settings. When this maximum value plateaus, it indicates those parameters settings and runtime have high power and are sufficient to recover the anomalous shapes.  

This process is much noisier than just measuring runtime as a function of parameter settings.  So in plots we show many data points and fit a trend line using a local average.  In the same plots, different scanning shapes naturally converge to distinct Measured Discrepancy values.  

For all of these experiments we use a rate of $p=.5$ for the data outside of the planted region and a rate of $q=.8$ for the data inside of the region.

\myParagraph{Flux Scanning Power.}  
The time required to achieve statistical power for the flux scanning model is shown in Figure \ref{fig:flux+part-err}.  
It shows the Measured Discrepancy values for disk, rectangle, and halfspace algorithms under the flux model on the Beijing data.  
We fix the planted shape to contain $5\%$ of the data ($f=0.05$).  As observed earlier, the disk regions are very slow to scan for under this model, and do not achieve high power until $2$ hours of runtime.  
However, the rectangle and halfspace scanning algorithms converge to a high power setting in about $10$ seconds to $1$ minute -- hence for the flux model, we recommend these shapes.  

\myParagraph{Partial Scanning Power.}  
Figure \ref{fig:flux+part-err} also shows the time require to achieve high statistical power under the partial model on OSM data.  We plant small ranges of size $f = 0.005$.
Each scanning shape (rectangle, disk, halfplane) is in a separate chart.  They all eventually achieve high statistical power, but the halfplane scanning algorithms only take about $1-2$ seconds, while the disk and rectangle algorithms require about $30$ minutes.  We can also observe that \Fme spatial approximation consistently converges faster than Random Sampling.  

\myParagraph{Full Scanning Power.}  
Finally, Figure \ref{fig:full-err} shows how long it takes to achieve high statistical power under the full model.  We plant regions in the OSM data with size $f = 0.005$.  We show the results of scaling runtime by varying $1/\eps$ (setting $1/\alpha = 6000$) and varying $1/\alpha$ (setting $1/\eps = 200$).  
For halfplanes, we use \Fmh which has no spatial error, so it only has its runtime vary as a function of $1/\eps$; and it achieves high power after about $1$ hour.  
For rectangles, it achieves high power through larger $1/\eps$ values in about $10$ minutes in both $1/\eps$ and $1/\alpha$ scaling.  
For radius-restricted disks, we show the MultiDisk scanning with \Fmgd (and \textsf{Hull Trick}); it appears less tied to the choice of $\alpha$ as this setting is tied to the resolution of the grid approximation.  For a fixed $1/\eps = 200$, it sometimes, but does not consistently find a high score region indicating that the setting $1/\eps = 200$ might be too small, but as $1/\eps$ increases after about 10 minutes it achieves high power.

\begin{figure}
	\centering
	\begin{tabular}{lcc} 
		& Inverse Spatial Error & Inverse Statistical Error \\ 
		\rotatebox[origin=c]{90}{Measured Discrepancy} & 	\hspace{-4mm}\raisebox{-0.5\height}{\includegraphics[width=0.522\linewidth]{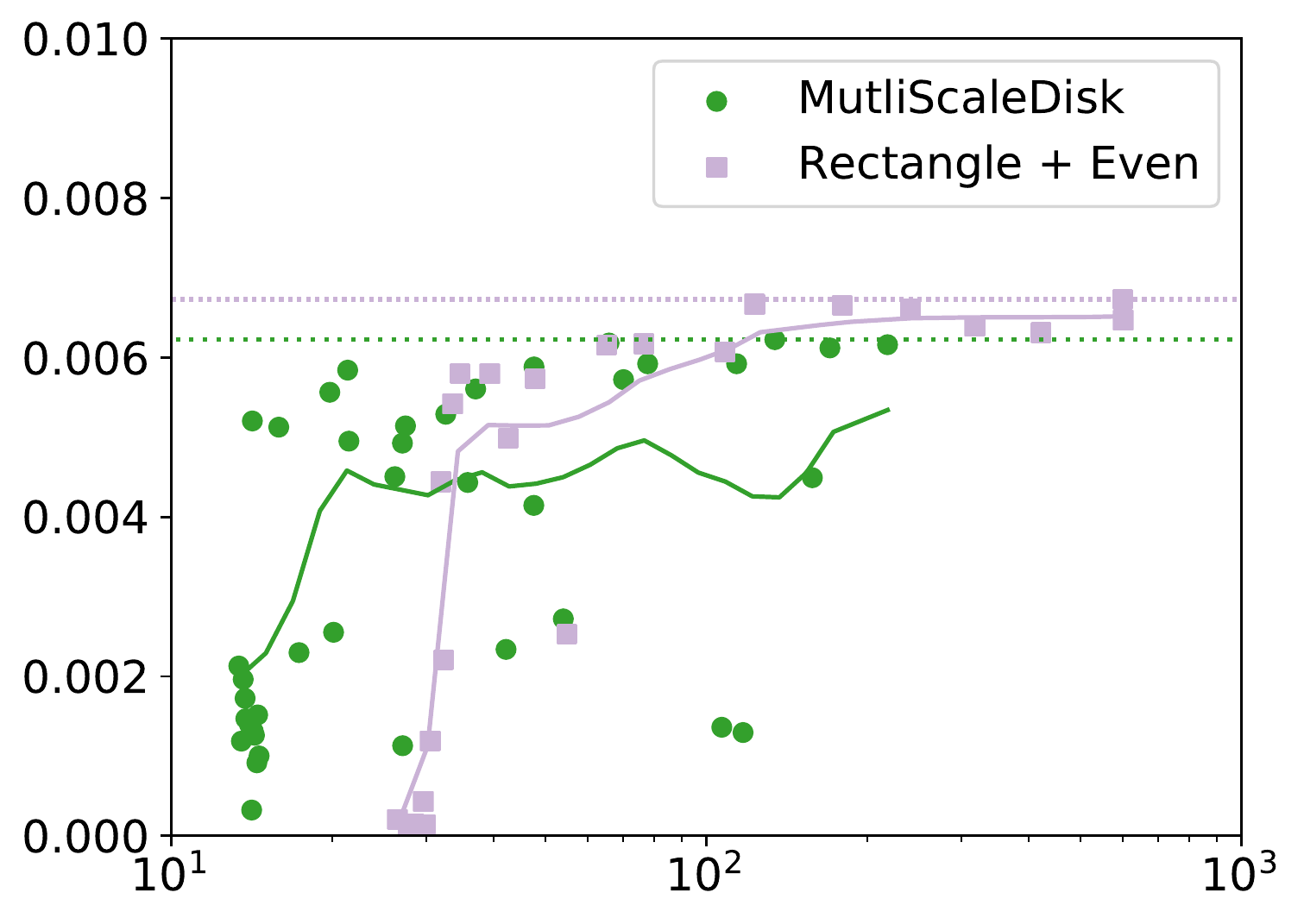}} & \hspace{-5mm}	\raisebox{-0.5\height}{\includegraphics[width=0.453\linewidth]{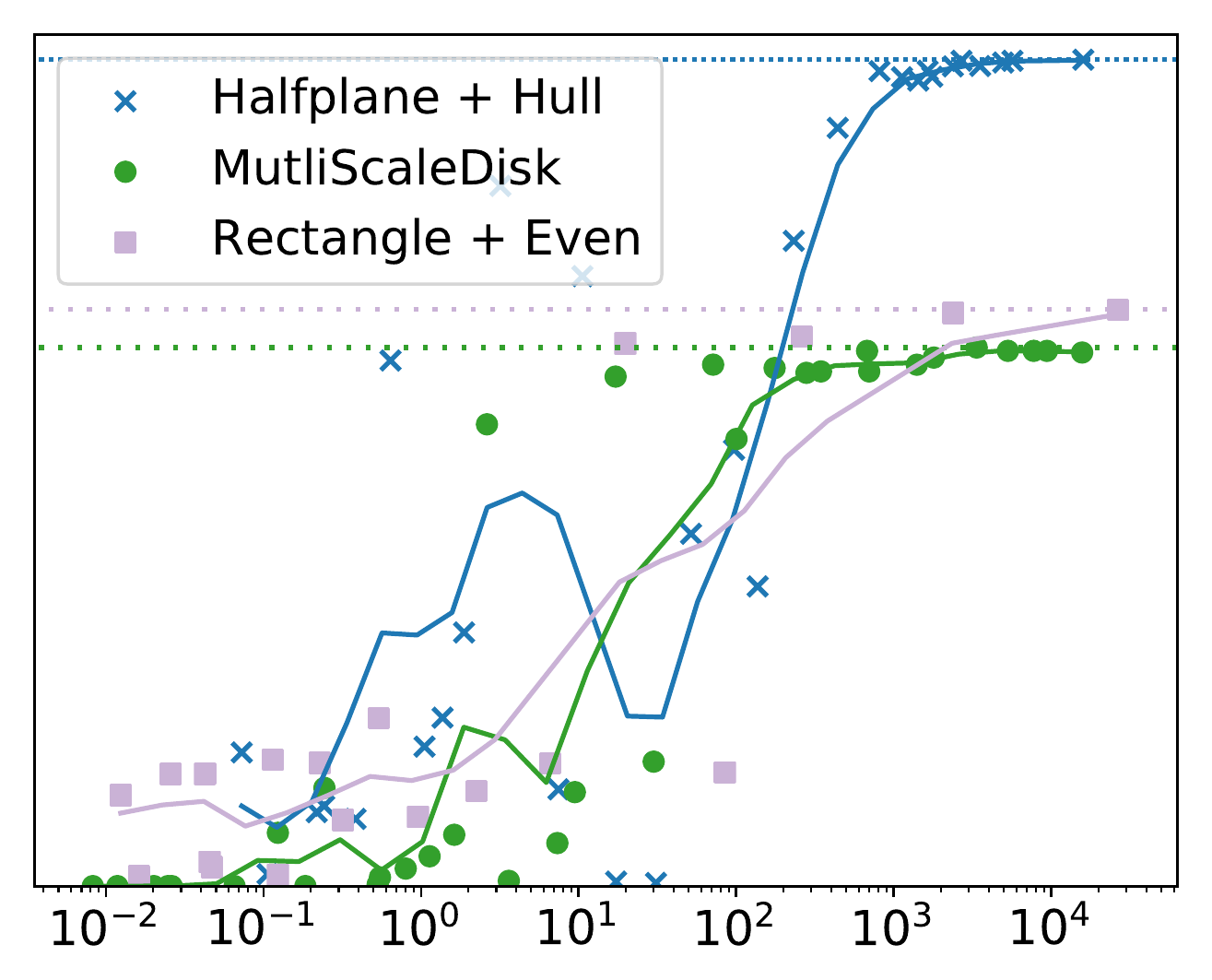}} \\ 
		& Time (sec) & Time (sec) \\
	\end{tabular}

	\caption{\label{fig:full-err}
	Statistical Power on OSM data with full model.}  
\end{figure}

\section{Previous and Related Work}
\label{sec:related}
Our algorithms build upon the recent two-level sampling framework for spatial scan statistics by Matheny \etal~\cite{SSSS,MP18b}.  This focused on making scalable $\eps$-approximate spatial scan statistics over point sets.  
This line of work provides improvements over non-approximate variants~\cite{Kul97,Kul7.0}, in terms of the number samples, and 
the runtime of the scanning algorithms for the same shapes we study: halfspaces, disks, and rectangles.   The introduction of two-level sample makes these approaches tractable, and better coresets or faster scanning makes them extremely efficient.  
However, these fast SSS methods only apply for point sets.

There exist other mechanisms for finding anomalous behavior among trajectories; these involves clustering by density and then identifying outliers, or training learning models on pre-labeled data~\cite{MT08, Dee2008, HXFXTM06, PMF08, YKR2008}. 
These approaches do not specifically identify spatial regions from characteristics of the trajectories or compare against a background population.  

A few papers have attempted to port scan statistics to trajectories, with goal of finding spatial anomalous regions.  
Pang \etal~\cite{PCLZ2011} discretized cities into grids and recorded traffic conditions in each cell. They then computed a likelihood ratio test over all sub-grids to detect the most anomalous region.  
And Liu \etal~\cite{LZCYX11} partitions a city into regions defined by the road network, and their adjacency defines a graph.  Then spatial anomalous links in this graph are scanned over various times to find a time$\times$region of high traffic.  
But neither of these approaches directly operate on the trajectories, and fix regions ahead of time which restricts the set and nature of possible anomalies.

An alternative approach (for point sets) removes the notion of shapes, and focuses on  clustering the measured objects to find potential anomalies~\cite{DA04,DCTB2007,Patil2004,Tango2005,OJPHI7599}.  However, this approach already does not have guarantees about statistical accuracy for point sets, and the task of clustering trajectories appropriately (and in this case one should be concerned about not over-fitting) has its own set of challenges, which are beyond the scope of this paper.

\section{Conclusion}

We introduce three new models for quantifying spatial anomalies among large sets of trajectories using spatial scan statistics and defined by geometric shapes.  These identify regions which exhibit high flux, a large percentage of the total arclength of a measured quantity, or a large percentage of a set of trajectories of interest pass through that region.  These models have numerous applications in traffic analysis, disease outbreak monitoring, epidemiology, and demography.  

Through either combinatorial, geometric reductions, or new scanning algorithms, we are able to identify these anomalous regions efficiently on data sets containing millions of trajectories with billions of waypoints.  This efficiency requires various insights tuned to the families of shapes we considered which include halfplanes, disks, and rectangles.  This includes careful ways to approximate trajectories by point sets, and fast enumeration methods.  These approximations are backed by a theoretical analysis, which shows guaranteed bounded error in the spatial trajectory approximation ($\alpha$) as well as in the measurement of the statistical quantities ($\eps$).  The runtime depends only on these parameters.  

And most importantly, we also measure the statistical power of the scanning algorithms.  That is, if we plant an anomalous region under each of the models, our scanning algorithms can with high probability, recover that region (or a similarly anomalous one) in tractable amounts of time.  Even on the millions of trajectories, high statistical power is often achieved within minutes, or for more challenging variants, in hours.

\clearpage
\bibliographystyle{plain}
\bibliography{discrepancy}

\end{document}